\newif\ifSpringer
\Springerfalse
\newif\ifBook
\Bookfalse
\newcommand{\chapitre}[2]{
\ifBook
\chapter[#1]{#2}
\else
\section[#1]{#2}
\fi}
\newcommand{\ssection}[2]{
\ifBook
\section[#1]{#2}
\else
\subsection[#1]{#2}
\fi}
\newcommand{\sssection}[2]{
\ifBook
\subsection[#1]{#2}
\else
\subsubsection{#2}
\fi}
\newcommand{\ssssection}[1]{
\ifBook
\subsubsection{#1}
\else
\paragraph{#1.}
\fi}

\newcommand{\bookarticle}{
\ifBook
book
\else
article
\fi}

\newcommand{\bookarticlepoint}{
\ifBook
book.
\else
article.
\fi}

\newcommand{\booksecvirg}{
\ifBook
book,
\else
article,
\fi}

\newcommand{\booksecpoint}{
\ifBook
book.
\else
article.
\fi}
\newcommand{\Chapsec}{
\ifBook
Chapter
\else
Section 
\fi}

\newcommand{\chaptersection}{
\ifBook
chapter
\else
section 
\fi}

\newif\ifElsev
\Elsevfalse
\ifElsev
\documentclass[preprint,authoryear]{elsarticle}
\else
\ifSpringer
\documentclass{svjour3}
\else
\ifBook
\documentclass{book}
\else
\documentclass{article}
\fi
\fi
\fi

\newif\ifcomment
\commenttrue 

\title{Rationality and Escalation\\ in Infinite Extensive Games}

\ifSpringer
\author{Pierre Lescanne\thanks{email: \textsf{Pierre.Lescanne@ens-lyon.fr}, phone: +33472728683, fax: +33472728080}}  
\institute{Universit\'e de Lyon, ENS de Lyon, CNRS (LIP), \\ 46 all\'ee
d'Italie, 69364 Lyon, France}
\else
\author{Pierre Lescanne\thanks{email: \textsf{Pierre.Lescanne@ens-lyon.fr}, phone: +33472728683, fax: +33472728080}\\ 
Universit\'e de Lyon, ENS de Lyon, CNRS (LIP), \\ 46 all\'ee
d'Italie, 69364 Lyon, France}
\date{}
\fi

\usepackage{textcomp} 
\usepackage{pstricks} 
\usepackage{pst-node} 
\usepackage{prooftree} 
\usepackage{coqdoc} 
\usepackage{url} 
\usepackage{xy}
\usepackage{qsymbols}
\usepackage{natbib}
\usepackage{makeidx}
\usepackage{lscape}
\usepackage{fancybox}
\usepackage{graphics}

\makeindex
\DeclareMathAlphabet{\mathpzc}{OT1}{pzc}{m}{it}

\newcount\timehh\newcount\timemm \timehh=\time
\divide\timehh by 60 \timemm=\time
\newcommand{\timestamp}{
  {\protect\small\sl\today\ --
    \ifnum\timehh<10 0\fi\number\timehh\,:\,
    \ifnum\timemm<10 0\fi\number\timemm}}

\ifElsev 
\newproof{proof}{Proof} 
\else 
\newenvironment{proof}[1]{\begin{quotation}\noindent\textsf{Proof:} #1}%
{\(\Box\)\end{quotation}}
\fi
\newtheorem{prop}{Proposition}
\newcommand{\Coq}{\textsc{Coq}} 
\newcommand{\ie}{i.e.,~} 
\newcommand{\leut}{\ensuremath{\preceq}} 

\newcommand{\hbis}{\sim_{h}}
\newcommand{\sbis}{\sim_{s}}
\newcommand{\gbis}{\mathop{\sim_{g}}}
\newcommand{\lft}{\ensuremath{\ell}}
\newcommand{\ltl}{\textit{``leads to a leaf''}}
\newcommand{\altl}{\textit{``always leads to a leaf''}}
\newcommand{\ob}{\mathfrak{o}}
\newcommand{\nat}{\ensuremath{\mathbb{N}}}
\newcommand{\Alice}{\textsc{Alice}}
\newcommand{\Bob}{\textsc{Bob}}
\newcommand{\alice}{\textsc{\scriptsize Alice}}
\newcommand{\bob}{\textsc{\scriptsize Bob}}
\newcommand{\rgt}{\ensuremath{\mathpzc{r}}}
\newcommand{\gs}{\textbf{g}}

\newcommand{\coconva}{\ensuremath{\mathop{"<<-"\!
\raisebox{1.5pt}{\(\scriptstyle a\)}\! "->>"}}} 
\newcommand{\conva}{\ensuremath{\mathop{\vdash\! \raisebox{2pt}{\(\scriptstyle a\)}\! \dashv}}} 
\newcommand{\convalice}{\ensuremath{\mathop{\vdash\! \raisebox{2pt}{\(\scriptstyle \Alice\)}\! \dashv}}}
\renewcommand{\H}{\ensuremath{\mathcal{H}}}
\newcommand{\og}{\ensuremath{\ll\!}}
\newcommand{\fg}{\ensuremath{\hspace*{-1.5pt}\gg}}
\newcommand{\ogame}{\ensuremath{\langle\!|}}
\newcommand{\fgame}{\ensuremath{|\!\rangle}}

\newcommand{\zone}{\textsf{zero\_one}}

\usepackage{color} %
\definecolor{darkbrown}{cmyk}{.3,.75,.75,.15}
\definecolor{vertfonce}{rgb}{0,.5,0}

\newcommand{\rouge}[1]{{\color{red} #1}}

\definecolor{vertfonce}{rgb}{0,.5,0}

\renewcommand{\coqdocleftpageheader}{\thepage\ -- }
\renewcommand{\coqdocrightpageheader}{\ -- \thepage}

\newcommand{\ziginf}
{
  \begin{figure}[thb]
    \centering
    \begin{center}
\(\psframe[linewidth=2pt,framearc=.33](.4,-.3)(1.2,.3)
\)
  \begin{psmatrix}[rowsep=6pt,colsep=1.8pt] 
    &&&&[name=z]{$\bullet$}\\
   &&&[name=a]{$\bullet$} \\
      &&[name=b]{$`()$} && [name=c]{$\bullet$} \\%
      &&&[name=d]{\textsf{zig}} && [name=e]{$`()$}
      \ncline{a}{b} \ncline{a}{c} %
      \ncline{c}{d} \ncline{c}{e} %
      \ncline[linestyle=dotted]{z}{a}
\end{psmatrix}
\qquad
\raisebox{10pt}{{\huge $"=>"$}}
\qquad
  \(\psframe[linewidth=2pt,framearc=.33](0,-.3)(2,1.5)\)
\begin{psmatrix}[rowsep=6pt,colsep=1.8pt]
    &&&&[name=z]{$\bullet$}\\
    &&&[name=a]{$\bullet$} \\
      &&[name=b]{$`()$} && [name=c]{$\bullet$} \\%
      &&&[name=d]{\textsf{zig}} && [name=e]{$`()$}
      \ncline{a}{b} \ncline{a}{c} %
      \ncline{c}{d} \ncline{c}{e} %
      \ncline[linestyle=dotted]{z}{a}
\end{psmatrix} %
\end{center}

\begin{center}
  \rule{200pt}{1pt}
\end{center}

\begin{center}
  \(\psframe[linewidth=2pt,framearc=.33](-.1,-.2)(.7,.3)\)
  \textsf{zig}
\end{center}
    \caption{How cofix works on \textsf{zig} for \textit{is infinite}?}
    \label{fig:ziginf}
  \end{figure}
}
\newcommand{\figdollar}
{
\begin{figure}[tb]
  \centering
  \hspace*{-30pt} 
\(\begin{psmatrix}[colsep=15pt,rowsep=20pt]
& [name=o]
  &{\ovalnode{a}{\Alice}} &{\ovalnode{b}{{\scriptstyle \Bob}}} & {\ovalnode{a1}{\Alice}} &{\ovalnode{b1}{{\scriptstyle \Bob}}}  & [name=c] & [name=d]
  &\\
  &[name=p]\phantom{\scriptstyle v+n, n} 
  &[name=e] {\scriptstyle v+n,n} 
  &[name=f] {\scriptstyle n+1,v+n} %
  &[name=e1] {\scriptstyle v+n+1,n+1} 
  &[name=f1] {\scriptstyle n+2,v+n+1} %
  &[name=h] \phantom{\scriptstyle 2n+1, 2n+2} 
  \ncarc[arrows=->,linestyle=dotted]{o}{p} %
  \ncarc[arrows=->,linestyle=dotted]{o}{a} %
  \ncarc[arrows=->]{a}{b} %
  \ncarc[arrows=->]{b}{a1} %
  \ncarc[arrows=->]{a1}{b1} %
  \ncarc[arrows=->,linestyle=dotted]{b1}{c} %
  \ncarc[arrows=->,linestyle=dotted]{c}{d} %
  \ncarc[arrows=->]{a}{e} %
  \ncarc[arrows=->]{b}{f} %
  \ncarc[arrows=->]{a1}{e1} %
  \ncarc[arrows=->]{b1}{f1} %
  \ncarc[arrows=->,linestyle=dotted]{c}{h} %
\end{psmatrix}
\)
  \caption{The \emph{dollar auction} game}
  \label{fig:dol_auct}
\end{figure}
}
\newcommand{\figfourstrat}
{\begin{figure}[thbp]
  \centering

    \hspace*{-30pt} 
\(\begin{psmatrix}[colsep=20pt,rowsep=20pt]
  &{\ovalnode{a}{{\scriptstyle \Alice}}} &{\ovalnode{b}{{\scriptstyle \Bob}}} & {\ovalnode{a1}{{\scriptstyle \Alice}}} &{\ovalnode{b1}{{\scriptstyle \Bob}}}  & [name=c] & [name=d]
  &\\
  &[name=e] {\scriptstyle v+n,n} 
  &[name=f] {\scriptstyle n+1,v+n} %
  &[name=e1] {\scriptstyle v+n+1,n+1} 
  &[name=f1] {\scriptstyle n+2,v+n+1} %
  &[name=h] \phantom{\scriptstyle n+1,v+n} 
  \ncarc[arrows=->]{a}{b} %
  \ncarc[arrows=->]{b}{a1} %
  \ncarc[arrows=->]{a1}{b1} %
  \ncarc[arrows=->,linestyle=dotted]{b1}{c} %
  \ncarc[arrows=->,linestyle=dotted]{c}{d} %
  \ncarc[arrows=->,linewidth=.08]{a}{e} %
  \ncarc[arrows=->,linewidth=.08]{b}{f} %
  \ncarc[arrows=->,linewidth=.08]{a1}{e1} %
  \ncarc[arrows=->,linewidth=.08]{b1}{f1} %
  \ncarc[arrows=->,linewidth=.08,linestyle=dotted]{c}{h} %
\end{psmatrix}
\)

\bigskip

\textsf{dolAsBs}$_n$ aka \textbf{Always give up}

\bigskip

  \hspace*{-30pt} 
\(\begin{psmatrix}[colsep=20pt,rowsep=20pt]
  &{\ovalnode{a}{{\scriptstyle \Alice}}} &{\ovalnode{b}{{\scriptstyle \Bob}}} & {\ovalnode{a1}{{\scriptstyle \Alice}}} &{\ovalnode{b1}{{\scriptstyle \Bob}}}  & [name=a2] & [name=b2]
  &\\
  &[name=e] {\scriptstyle v+n,n} 
  &[name=f] {\scriptstyle n+1,v+n} %
  &[name=e1] {\scriptstyle v+n+1,n+1} 
  &[name=f1] {\scriptstyle n+2,v+n+1} %
  &[name=h] \phantom{\scriptstyle n+1,v+n} 
  \ncarc[arrows=->]{a}{b} %
  \ncarc[arrows=->,linewidth=.08]{a}{e} %
  \ncarc[arrows=->,linewidth=.08]{b}{a1} %
  \ncarc[arrows=->]{b}{f} %
  \ncarc[arrows=->]{a1}{b1} %
  \ncarc[arrows=->,linewidth=.08]{a1}{e1} %
  \ncarc[arrows=->,linestyle=dotted,linewidth=.08]{b1}{a2} %
  \ncarc[arrows=->]{b1}{f1} %
  \ncarc[arrows=->,linestyle=dotted]{a2}{b2} %
  \ncarc[arrows=->,linestyle=dotted,linewidth=.08]{a2}{h} %
\end{psmatrix}
\)

\bigskip

\textsf{dolAsBc}$_n$ aka \textbf{\Alice{} abandons always and \Bob{} continues always}

\bigskip

  \hspace*{-30pt} 
\(\begin{psmatrix}[colsep=20pt,rowsep=20pt]
  &{\ovalnode{a}{{\scriptstyle \Alice}}} &{\ovalnode{b}{{\scriptstyle \Bob}}} & {\ovalnode{a1}{{\scriptstyle \Alice}}} &{\ovalnode{b1}{{\scriptstyle \Bob}}}  & [name=a2] & [name=b2]
  &\\
  &[name=e] {\scriptstyle v+n,n} 
  &[name=f] {\scriptstyle n+1,v+n} %
  &[name=e1] {\scriptstyle v+n+1,n+1} 
  &[name=f1] {\scriptstyle n+2,v+n+1} %
  &[name=h] \phantom{\scriptstyle n+1,v+n} 
  \ncarc[arrows=->,linewidth=.08]{a}{b} %
  \ncarc[arrows=->,linewidth=.02]{a}{e} %
  \ncarc[arrows=->,linewidth=.02]{b}{a1} %
  \ncarc[arrows=->,linewidth=.08]{b}{f} %
  \ncarc[arrows=->,linewidth=.08]{a1}{b1} %
  \ncarc[arrows=->,linewidth=.02]{a1}{e1} %
  \ncarc[arrows=->,linestyle=dotted,linewidth=.02]{b1}{a2} %
  \ncarc[arrows=->,linewidth=.08]{b1}{f1} %
  \ncarc[arrows=->,linestyle=dotted,linewidth=.08]{a2}{b2} %
  \ncarc[arrows=->,linestyle=dotted,linewidth=.02]{a2}{h} %
\end{psmatrix}
\)

\bigskip

\textsf{dolAcBs}$_n$ aka \textbf{\Alice{} continues always and \Bob{} abandons always}

  \caption{Three strategy profiles}
  \label{fig:4_strat}
\end{figure}
}
\newcommand{\figzeroone}
{
  \begin{figure}[ht!]
    \centering \hspace*{-30pt} \(
    \begin{psmatrix}[colsep=20pt,rowsep=20pt] &
      {\ovalnode{a}{\alice}} 
      &{\ovalnode{b}{\bob}} %
      &[name=c]
      & [name=d] 
      &{\ovalnode{e}{\alice}}
      &{\ovalnode{f}{\bob}}
      & {\ovalnode{g}{\alice}}
      & [name=h] 
      & [name=i]
      &\\ 
      &[name=a1] {\scriptstyle 0,1} %
      &[name=b1] {\scriptstyle 1,0} %
      & [name=c1]\phantom{\scriptstyle 1,0} %
      &[name=d1] {\scriptstyle 0,1} %
      &[name=e1] {\scriptstyle 1,0} %
      &[name=f1] {\scriptstyle 0,1} %
      &[name=g1] \phantom{\scriptstyle 1, 0} %
      & [name=h1] %
      \ncarc[arrows=->]{a}{b}\Aput{\scriptstyle \mathsf{c}} %
      \ncarc[arrows=->,linestyle=dotted]{b}{c} \Aput{\scriptstyle \mathsf{c}} %
      \ncarc[arrows=->,linestyle=dotted]{d}{e} \Aput{\scriptstyle \mathsf{c}} %
      \ncarc[arrows=->]{e}{f} \Aput{\scriptstyle \mathsf{c}} %
      \ncarc[arrows=->]{f}{g} \Aput{\scriptstyle \mathsf{c}} %
      \ncarc[arrows=->,linestyle=dotted]{g}{h}\Aput{\scriptstyle \mathsf{c}} %
      \ncarc[arrows=->,linestyle=dotted]{h}{i} \Aput{\scriptstyle \mathsf{c}} %
       \ncarc[arrows=->]{a}{a1} \Aput{\scriptstyle \mathsf{s}}%
       \ncarc[arrows=->]{b}{b1} \Aput{\scriptstyle \mathsf{s}} %
      \ncarc[arrows=->,linestyle=dotted]{d}{d1} \Aput{\scriptstyle \mathsf{s}} %
      \ncarc[arrows=->]{e}{e1} \Aput{\scriptstyle \mathsf{s}} %
      \ncarc[arrows=->]{f}{f1} \Aput{\scriptstyle \mathsf{s}} %
      \ncarc[arrows=->]{g}{g1}\Aput{\scriptstyle \mathsf{s}}  %
      \ncarc[arrows=->,linestyle=dotted]{h}{h1} \Aput{\scriptstyle \mathsf{s}}%
    \end{psmatrix}
    \)
    \caption{The $0,1$ game.}
    \label{fig:zone}
  \end{figure}
}

\newcommand{\rosenthal}{
  \begin{figure}[ht]
\hspace*{-20pt}
    \(\begin{psmatrix}[colsep=15pt,rowsep=15pt]
  &{\ovalnode{un}{A}} &{\ovalnode{de}{B}} & {\ovalnode{tr}{A}} &{\ovalnode{qu}{B}}  
    &{\ovalnode{ci}{A}} &{\ovalnode{si}{B}} & {\ovalnode{se}{A}} &{\ovalnode{hu}{B}} 
   &{\ovalnode{ne}{A}} &{\ovalnode{di}{B}} & [name=on]{\scriptstyle (0,10)}
 &\\
  &[name=un1] {\scriptstyle (0,0)} 
  &[name=de1] {\scriptstyle (-1,3)} %
  &[name=tr1] {\scriptstyle (2,2)} 
  &[name=qu1] {\scriptstyle (1,5)} %
  &[name=ci1]  {\scriptstyle (4,4)} %
  &[name=si1]  {\scriptstyle (3,7)} %
  &[name=se1]  {\scriptstyle (6,6)} %
  &[name=hu1]  {\scriptstyle (5,9)} %
  &[name=ne1]  {\scriptstyle (8,8)} %
  &[name=di1]  {\scriptstyle (7,11)} %
  \ncarc[arrows=->,linewidth=.02]{un}{de} %
  \ncarc[arrows=->,linewidth=.02]{un}{un1} %
  \ncarc[arrows=->,linewidth=.02]{de}{tr} %
  \ncarc[arrows=->,linewidth=.02]{de}{de1} %
  \ncarc[arrows=->,linewidth=.02]{tr}{qu} %
  \ncarc[arrows=->,linewidth=.02]{tr}{tr1} %
  \ncarc[arrows=->,linewidth=.02]{qu}{ci} %
  \ncarc[arrows=->,linewidth=.02]{qu}{qu1} %
  \ncarc[arrows=->,linewidth=.02]{ci}{si} %
  \ncarc[arrows=->,linewidth=.02]{ci}{ci1} %
  \ncarc[arrows=->,linewidth=.02]{si}{se} %
  \ncarc[arrows=->,linewidth=.02]{si}{si1} %
  \ncarc[arrows=->,linewidth=.02]{se}{hu} %
  \ncarc[arrows=->,linewidth=.02]{se}{se1} %
  \ncarc[arrows=->,linewidth=.02]{hu}{ne} %
  \ncarc[arrows=->,linewidth=.02]{hu}{hu1} %
  \ncarc[arrows=->,linewidth=.02]{ne}{di} %
  \ncarc[arrows=->,linewidth=.02]{ne}{ne1} %
 \ncarc[arrows=->,linewidth=.02]{di}{on} %
  \ncarc[arrows=->,linewidth=.02]{di}{di1} %
\end{psmatrix}
\)
    \caption{The genuine Rosenthal game}
    \label{fig:rosen}
  \end{figure}
}

\newcommand{\infinipede}{
  \begin{figure}[ht]
    \centering
    \(\begin{psmatrix}[colsep=10pt,rowsep=20pt]
  &{\ovalnode{un}{A}} &{\ovalnode{de}{B}} & {\ovalnode{tr}{A}} &[name=qu]
    &[name=ci] &{\ovalnode{si}{A}} & {\ovalnode{se}{B}} &{\ovalnode{hu}{A}} 
   &[name=ne]&[name=di]
 &\\
  &[name=un1] {\scriptstyle \quad (0,0)\quad } 
  &[name=de1] {\scriptstyle \quad (0,3)\quad } %
  &[name=tr1] {\scriptstyle \quad (2,2)\quad } 
  &[name=qu1] {\scriptstyle \qquad} %
  &[name=ci1]  
  &[name=si1]  {\scriptscriptstyle \ (2n,2n)\ } %
  &[name=se1]  {\scriptscriptstyle (2n-1,2n+3)} %
  &[name=hu1]  {\scriptscriptstyle (2n+2,2n+2)} %
  &[name=ne1]  
  &[name=di1]  
  \ncarc[arrows=->,linewidth=.02]{un}{de} %
  \ncarc[arrows=->,linewidth=.02]{un}{un1} %
  \ncarc[arrows=->,linewidth=.02]{de}{tr} %
  \ncarc[arrows=->,linewidth=.02]{de}{de1} %
  \ncarc[arrows=->,linewidth=.02,linestyle=dotted]{tr}{qu} %
  \ncarc[arrows=->,linewidth=.02]{tr}{tr1} %
   \ncarc[arrows=->,linewidth=.02,linestyle=dotted]{ci}{si} %
  \ncarc[arrows=->,linewidth=.02]{si}{se} %
  \ncarc[arrows=->,linewidth=.02]{si}{si1} %
  \ncarc[arrows=->,linewidth=.02]{se}{hu} %
  \ncarc[arrows=->,linewidth=.02]{se}{se1} %
  \ncarc[arrows=->,linewidth=.02,linestyle=dotted]{hu}{ne} %
  \ncarc[arrows=->,linewidth=.02]{hu}{hu1} %
\end{psmatrix}
\)
    \caption{The infinipede}
    \label{fig:infinipede}
  \end{figure}
}

\newcommand{\infSGPE}{
  \begin{figure}[ht]
    \centering
    \(\begin{psmatrix}[colsep=10pt,rowsep=20pt]
  &{\ovalnode{un}{A}} &{\ovalnode{de}{B}} & {\ovalnode{tr}{A}} &[name=qu]
    &[name=ci] &{\ovalnode{si}{A}} & {\ovalnode{se}{B}} &{\ovalnode{hu}{A}} 
   &[name=ne]&[name=di]
 &\\
  &[name=un1] {\scriptstyle \quad (0,0)\quad } 
  &[name=de1] {\scriptstyle \quad (0,3)\quad } %
  &[name=tr1] {\scriptstyle \quad (2,2)\quad } 
  &[name=qu1] {\scriptstyle \qquad} %
  &[name=ci1]  
  &[name=si1]  {\scriptscriptstyle \ (2n,2n)\ } %
  &[name=se1]  {\scriptscriptstyle (2n-1,2n+3)} %
  &[name=hu1]  {\scriptscriptstyle (2n+2,2n+2)} %
  &[name=ne1]  
  &[name=di1]  
  \ncarc[arrows=->,linewidth=.02]{un}{de} %
  \ncarc[arrows=->,linewidth=.08]{un}{un1} %
  \ncarc[arrows=->,linewidth=.02]{de}{tr} %
  \ncarc[arrows=->,linewidth=.08]{de}{de1} %
  \ncarc[arrows=->,linewidth=.02,linestyle=dotted]{tr}{qu} %
  \ncarc[arrows=->,linewidth=.08]{tr}{tr1} %
   \ncarc[arrows=->,linewidth=.02,linestyle=dotted]{ci}{si} %
  \ncarc[arrows=->,linewidth=.02]{si}{se} %
  \ncarc[arrows=->,linewidth=.08]{si}{si1} %
  \ncarc[arrows=->,linewidth=.02]{se}{hu} %
  \ncarc[arrows=->,linewidth=.08]{se}{se1} %
  \ncarc[arrows=->,linewidth=.02,linestyle=dotted]{hu}{ne} %
  \ncarc[arrows=->,linewidth=.08]{hu}{hu1} %
\end{psmatrix}
\)
    \caption{Subgame perfect equilibrium for the infinipede}
    \label{fig:infSGPE}
  \end{figure}
}

\newcommand{\AliceBob}{
\begin{figure}[!ht]
\vspace*{5cm}
\hspace*{3cm}
\mbox{
\psset{xunit=.4cm}\psset{yunit=.4cm}
\pscircle[linewidth=.5pt,linecolor=gray](-.8,4){.1}
\pscircle[linewidth=.5pt,linecolor=gray](-1.5,5){.1}
\pscircle[linewidth=.5pt,linecolor=gray](-2,6){.1}
\pscircle[linewidth=.5pt,linecolor=gray](-2.5,7){.1}

\qquad \rput*(-1,9){\large \sf He bluffs!
\(
\begin{psmatrix}[colsep=8pt,rowsep=3pt]
  &{\ovalnode{a}{{\,}}} &{\ovalnode{b}{{\, }}} &
  {\ovalnode{a1}{{\, }}} &{\ovalnode{b1}{{\ }}}  & [name=a2] & [name=b2]
  &\\
  &[name=e] {\bullet}
  &[name=f] {\bullet}
  &[name=e1] {\bullet}
  &[name=f1] {\bullet}
  &[name=h] {\bullet}
  \ncarc[arrows=->,linewidth=.05]{a}{b} %
  \ncarc[arrows=->,linewidth=.02]{a}{e} %
  \ncarc[arrows=->,linewidth=.02]{b}{a1} %
  \ncarc[arrows=->,linewidth=.05]{b}{f} %
  \ncarc[arrows=->,linewidth=.05]{a1}{b1} %
  \ncarc[arrows=->,linewidth=.02]{a1}{e1} %
  \ncarc[arrows=->,linestyle=dotted,linewidth=.02]{b1}{a2} %
  \ncarc[arrows=->,linewidth=.05]{b1}{f1} %
  \ncarc[arrows=->,linestyle=dotted,linewidth=.05]{a2}{b2} %
  \ncarc[arrows=->,linestyle=dotted,linewidth=.02]{a2}{h} %
\end{psmatrix}
\)
}

\psellipse[linewidth=.2pt](1,0)(3,4)
\qdisk(0.1,1.7){1.5pt}
\parabola[linewidth=1pt](-.5, 1.7)(0, 1.55)
\parabola[linewidth=1pt](-.5, 1.7)(0, 1.85)
\parabola[linewidth=2pt](-.5, 2)(0, 2.15)
\qdisk(2.16,1.7){1.5pt}
\parabola[linewidth=1pt](1.5, 1.7)(2, 1.55)
\parabola[linewidth=1pt](1.5, 1.7)(2, 1.85)
\parabola[linewidth=2pt](1.5, 2)(2, 2.15)
\psarc[linewidth=1pt](1,0){.1}{260}{110}
\parabola(.1,-1.5)(1,-2)
\parabola(.1,-1.5)(1,-2.2)
\pscircle[linecolor=white,fillstyle=solid,fillcolor=lightgray](-.5,-.5){10pt}
\pscircle[linecolor=white,fillstyle=solid,fillcolor=lightgray](2.6,-.5){10pt}

\hspace*{1.2cm}

\psellipse[linewidth=.2pt](10, 0)(3, 3)
\parabola[linewidth=2pt](8.8, 1.7)(9, 1.6)
\parabola[linewidth=2pt](10.8, 1.7)(11, 1.6)
\parabola(8.6, -1.5)(9.5, -2)
\psarc(9.5, 0){.4}{30}{310}

\pscircle[linewidth=.5pt,linecolor=gray](12.8,3){.1}
\pscircle[linewidth=.5pt,linecolor=gray](13.1,4){.1}
\pscircle[linewidth=.5pt,linecolor=gray](13.4,5){.1}

\rput*(11,7){\large \sf She bluffs!
\(
\begin{psmatrix}[colsep=8pt,rowsep=3pt]
  &{\ovalnode{a}{{\,}}} &{\ovalnode{b}{{\, }}} &
  {\ovalnode{a1}{{\, }}} &{\ovalnode{b1}{{\ }}}  & [name=a2] & [name=b2]
  &\\
  &[name=e] {\bullet}
  &[name=f] {\bullet}
  &[name=e1] {\bullet}
  &[name=f1] {\bullet}
  &[name=h] {\bullet}
  \ncarc[arrows=->,linewidth=.05]{a}{b} %
  \ncarc[arrows=->,linewidth=.02]{a}{e} %
  \ncarc[arrows=->,linewidth=.02]{b}{a1} %
  \ncarc[arrows=->,linewidth=.05]{b}{f} %
  \ncarc[arrows=->,linewidth=.05]{a1}{b1} %
  \ncarc[arrows=->,linewidth=.02]{a1}{e1} %
  \ncarc[arrows=->,linestyle=dotted,linewidth=.02]{b1}{a2} %
  \ncarc[arrows=->,linewidth=.05]{b1}{f1} %
  \ncarc[arrows=->,linestyle=dotted,linewidth=.05]{a2}{b2} %
  \ncarc[arrows=->,linestyle=dotted,linewidth=.02]{a2}{h} %
\end{psmatrix}
\)
}
}

\vspace*{2.5cm}
\hspace*{2cm}
\hspace*{1.8cm} {\large \sf Alice} \hspace*{4cm} {\large \sf Bob}

  \caption{What Alice and Bob might think}
\label{fig:AliceBob}
\end{figure}
}
\begin{document} 
\newqsymbol{"|-|"}{\conva} 
\newqsymbol{"<<-a->>"}{\coconva}

\maketitle

\begin{abstract}
  The aim of this \bookarticle{} is to study infinite games and to prove formally properties
  in this framework.  In particular, we show that the behavior which leads to speculative
  crashes or escalation is  rational.  Indeed it proceeds logically  from the statement that
  resources are infinite. The reasoning is based on the concept of coinduction conceived
  by computer scientists to model infinite computations and used by rational agents
  unknowingly.  
\end{abstract}

\tableofcontents

\thispagestyle{empty}

\chapitre{Introduction}{Introduction} 

The aim of this\bookarticle{} is to study infinite games and to prove formally some
properties in this framework.  As a consequence, we show that the behavior (the madness)
of people which leads to speculative crashes or escalation can be proved rational.
Indeed it proceeds from the statement that resources are infinite. The reasoning is based
on the concept of coinduction conceived by computer scientists to model infinite
computations and used by rational agents unknowingly.  When used consciously, this concept
is not as simple as induction and we could paraphrase Newton \citep{bouchaud08:_econom}:
\emph{``Modeling the madness of people is more difficult than modeling the motion of
  planets''}.

In this \chaptersection{} we present the three words of the title, namely rationality,
escalation and  infiniteness.

\ssection{Rationality and escalation}{Rationality and escalation}

We consider the ability of agents to reason and to conduct their action according to a
line of reasoning.  We call this \emph{rationality}. This could have been called
\emph{wisdom} as this attributed to King Solomon.  It is not clear that agents act always
rationally.  If an agent acts always following a strict reasoning one says that he (she)
is rational.  To specify strictly this ability, one associates the agent with a mechanical
reasoning device, more specifically a Turing machine or a similar decision mechanism based
on abstract computations.  One admits that in making a decision the agent chooses the
option which is the better, that is no other will give better payoff, one says that this
option is an equilibrium in the sense of game theory.  A well-known game theory situation
where rationality of agents is questionable is the so-called \emph{escalation}. This is a
situation where there is a sequence of decisions wich can be infinite. If many
agents\footnote{Usually two are enough.} act one after the others in an infinite sequence
of decisions and if this sequence leads to situations which are worst and worst for the
agents, one speaks of \emph{escalation}.  One notices the emergence of a property of
complex systems, namely the behavior of the system is not the conjunction of this of all
the constituents.  Here the individual wisdom becomes a global madness.

\ssection{Infiniteness}{Infiniteness}
\label{sec:infiniteness}\index{infiniteness}

It is notorious that there is a wall between finiteness and infiniteness, a fact known to
model theorists like \citet{fagin93:_finit_model_theor_person_persp} \citet{EF-finite-mt}
and to specialists of  functions of real variable.  \citet{weierstrass72}  gave an
example of the fact that a finite sum of functions differentiable everywhere is
differentiable everywhere whereas an infinite sum is differentiable nowhere.  This
confusion between finite and infinite is at the origin of the conclusion of the
irrationality of  the escalation founded on  the belief
that a property of a infinite mathematical object can be extrapolated from a similar property of
finite approximations.\footnote{In the postface (\Chapsec~\ref{cha:postface}) we give
  another explanation: agents stipulate an a priori hypothesis that resources are finite
  and that therefore escalation is impossible.}
  As \citet{fagin93:_finit_model_theor_person_persp}
recalls, ``Most of the classical theorems of logic [for infinite structures] fail for
finite structures'' (see~\citet{EF-finite-mt} for a full development of the finite model
theory).  The reciprocal holds obviously: ``Most of the results which hold for finite
structures, fail for infinite structures''.  This has been beautifully evidenced in
mathematics, when \citet{weierstrass72} has exhibited his function:
\begin{displaymath} f(x)=\sum_{n=0}^\infty b^n\cos(a^n x \pi).
\end{displaymath} Every finite sum is differentiable and the limit, \ie the
\emph{infinite} sum, is not.  In another domain,  \citet{green08}\footnote{Toa won the
  Fields Medal for this work.} have
proved that the sequence of prime numbers contains arbitrarily long arithmetic
progressions.  By extrapolation, there would  exist an infinite
arithmetic progression of prime numbers, which is trivially not true. 
To give another picture, infinite games are to finite games
what fractal curves are to smooth curves \citep{edgar08:fract} and \citep[p.~174-175]{paulos03:_mathem_plays_stock_market}.  In game theory the error
done by the ninetieth century mathematicians \footnote{Weierstrass quotes some careless 
  mathematicians, namely Cauchy, Dirichlet and
  Gauss,  whereas Riemann was conscient of the problem.}  \index{Weierstrass}
would lead to the same issue.  With what we are concerned, a
result which holds on finite games does not hold necessarily on infinite games and
vice-versa.  More specifically equilibria on finite games are not preserved at the limit
on infinite games whereas new types of equilibria emerge on the infinite game not present
in the approximation (see the $0,1$ game in Section~\ref{sec:01}) and Section~\ref{sec:0-1-game}.
In particular, we cannot conclude that, whereas the only rational
attitude in finite dollar auction would be to stop immediately, it is irrational to
escalate in the case of an infinite auction.  We have to keep in mind that in the case of
escalation, the game is infinite, therefore reasoning made for finite objects are
inappropriate and tools specifically conceived for infinite objects should be adopted.
Like Weierstrass' discovery led to the development of function series, logicians have devised
methods for correct deductions on infinite structures. The right framework for
reasoning logically on infinite mathematical objects is called \emph{coinduction}\index{coinduction}.

The inadequate reasoning on infinite games is as follows: people study finite
approximation of infinite games as infinite games truncated at a finite location.  If they
obtain the same result on all the approximations, they extrapolate the result to the
infinite game as if the limit would have the same property.  But this says nothing since
the \emph{infiniteness is not the limit of finiteness}.  Instead of reexamining their
reasoning or considering carefully the hypotheses their reasoning is based upon, namely
wondering whether the set of resource is infinite, they conclude that humans are
irrational.  If there is an escalation, then the game is infinite, then \emph{the
  reasoning must be specific to infinite games, that is based on coinduction}.  This is
only on this basis that one can conclude that humans are rational or irrational.  In no
case, a property on the infinite game generated by escalation cane be extrapolated from
the same property on finite games. \index{escalation}

In this \bookarticle{} we address these issues.  The games we consider may have arbitrary long
histories as well as infinite histories.  \index{history} In our games there are two choices at each
node\footnote{Those choices are often to \emph{stop} or to \emph{continue}.},
this will not loose generality, since we can simulate finitely branching games in this
framework.  By K\"onig's lemma, finitely branching, specifically binary, infinite games
have at least an infinite history.  We are taking the problem of defining formally
infinite games, infinite strategy profiles, and infinite histories extremely seriously.
By ``seriously'' we mean that we prepare the land for precise, correct and rigorous reasoning.
For instance, an important issue which is not considered in the literature is how the utilities
associated with an infinite history are computed.  To be formal and rigorous, we expect
some kinds of recursive definitions, more precisely co-recursive definitions, but then
comes the questions of what the payoff associated with an infinite strategy profile is and
whether such a payoff exists (see Section~\ref{sec:why-infinite-plays}).

\ssection{Games}{Games}
\label{sec:games}

Finite extensive games are represented by finite trees and are analyzed through induction.
For instance, in finite extensive games, a concept like \emph{subgame perfect equilibrium}
is defined inductively and receives appropriately the name of \emph{backward induction}.
Similarly \emph{convertibility} (an agent changes choices in his strategy) has
also an inductive definition and this concept is a key for this of \emph{Nash
  equilibrium}.  But \emph{induction}, which has been designed for finitely based objects,
no more works on infinite\footnote{In this \booksecvirg \emph{infinite} means infinite and
  discrete.  For us, an infinite extensive game is discrete and has infinitely many nodes.
} games, \ie games underlying infinite trees.  Logicians have proposed a tool, which
they call \emph{coinduction}, to reason on infinite objects.  In short, since objects are
infinite and their construction cannot be analyzed, coinduction ``observes'' them, that is
looks at how they ``react'' to operations (see Section~\ref{sec:inf_obj} for more
explanation).  In this \booksecvirg we formalize with coinduction, the concept of infinite
game, of infinite strategy profile, of equilibrium in infinite games, of utility (payoff),
and of subgame.  We verify on the proof assistant \Coq{} that everything works smoothly
and yields interesting consequences.  Thanks to coinduction,  examples of apparently
paradoxical human behavior are explained logically, demonstrating a rational behavior.

Finite extensive games have been introduced by \cite{Kuhn:ExtGamesInfo53}.  But many
interesting extensive games are infinite and therefore the theory of infinite extensive
games play an important role in game theory, with examples like the \emph{dollar auction
  game} \index{Shubik}
\citep{Shubik:1971,colman99:_game_theor_and_its_applic,gintis00:_game_theor_evolv,osborne94:_cours_game_theory},
the generalized \emph{centipede game}\footnote{Here \emph{``generalized''} means that the
  game has an infinite ``backbone''.} or infinipede or the \emph{$0,1$ game}.  From a formal point of view, the concepts
associated with infinite extensive games are not appropriately treated in papers and
books. In particular, there is no clear notion of Nash equilibrium in infinite extensive
game and the gap between finiteness and infiniteness is not correctly understood.  For
instance in one of the textbooks on game theory, one finds the following definition of
\emph{games with finite horizon}:
\begin{it}
  If the length of the longest derivation is [...] finite, we say that the game has a
  finite horizon.
  Even a game with a finite horizon may have infinitely many terminal histories, because
  some player has infinitely many actions after some history.
\end{it}
\noindent
Notice that in an infinite game with infinite branching it is not always the case that a
longest derivation exists.  If a game has only finite histories, but has infinitely many
such finite histories of increasing length, there is no longest history.  Before giving a
formal definition later in the \booksecvirg let us say intuitively why this definition is
inconsistent.  Roughly speaking, a history is a path in the ordered tree which underlies
the game.  A counterexample 
is precisely when the
tree is infinitely branching \ie when ``some player has infinitely many
actions''.\footnote{In \Chapsec~\ref{sec:inf_stra}, we define a predicate \emph{leads to
    leaf} which we think to characterize properly the concept of \emph{finite horizon}
  which is a property of strategy profiles.}

Escalation takes place in specific sequential games in which players continue although
their payoff decreases on the whole.  The \emph{dollar auction} game \index{dollar auction} has been presented by
\citet{Shubik:1971} \index{Shubik} as the paradigm of escalation.  He noted that, even though their cost
(the opposite of the payoff) basically increases, players may keep bidding.  This attitude
was considered as inadequate and when talking about escalation, \citet{Shubik:1971} says
this is a paradox, \citet{oneill86:_inten_escal_and_dollar_auction} and
\citet{leininger89:_escal_and_coop_in_confl_situat} consider the bidders as irrational,
\citet{gintis00:_game_theor_evolv} speaks of \emph{illogic conflict of escalation} and
\citet{colman99:_game_theor_and_its_applic} calls it \textit{Macbeth effect} after
Shakespeare's play.  Rebutting these authors, we prove  in this \booksecvirg using a
reasoning conceived for infinite structures that escalation is logic and that agents are
rational, therefore this is not a paradox and we are led to assert that Macbeth is in
some way rational.

This escalation phenomenon occurs in infinite sequential games and only there.
\index{escalation}  To quote \citet{Shubik:1971}:
\begin{quotation}
  \begin{it}
    We could add an upper limit to the amount that anyone is
          allowed to bid.  However the analysis is confined to the (possibly infinite) game without
          a specific termination point, as no particularly interesting general phenomena appear if
          an upper bound is introduced.
  \end{it}
\end{quotation}
Therefore
it must be studied in infinite games with adequate tools, \ie in a framework designed for
mathematical infinite objects.  Like \citet{Shubik:1971} we will limit ourselves to two
players only.  In auctions, this consists in the two players bidding forever.  This
statement of rationality is based on the largely accepted assumption that a player is
rational if he adopts a strategy which corresponds to a \emph{subgame perfect
  equilibrium}.  To characterize this equilibrium most of the above cited authors consider a
finite restriction of the game for which they compute the subgame perfect equilibrium by
\emph{backward induction}\footnote{What is called ``backward induction'' in game theory is
  roughly what is called ``induction'' in logic.}.  Then they extrapolate the result
obtained on the amputated games to the infinite game.  To justify their practice, they add
a new hypothesis on the amount of money the bidders are ready to pay, called the
\emph{limited bankroll}.  By enforcing the finiteness of the game, they exclude clearly escalation.
In the amputated game dollar auction, they conclude that there
is a unique subgame perfect equilibrium.  This consists in both agents giving up
immediately, not starting the auction and adopting the same choice at each step.  In our
formalization in infinite games, we show that extending that case up to infinity is not a
subgame perfect equilibrium and we found two subgame perfect equilibria, namely the cases
when one agent continues at each step and the other leaves at each step.  Those equilibria
which correspond to rational attitudes account for the phenomenon of escalation.  Actually
this discrepancy between equilibrium in amputated games extrapolated to infinite extensions
and infinite games occurs in a much simpler game than the dollar auction namely the
$0,1$~game \index{$0,1$ game} which will be studied in this \booksecpoint

\ssection{Coinduction}{Coinduction}
\label{sec:coinduction} \index{coinduction}

Like induction, coinduction is based on a fixpoint, but whereas induction is based on the
least fixpoint, coinduction is based on the greatest fixpoint, for an ordering that we are not
going to describe here, since it would go beyond the scope of this \booksecpoint  Attached to
induction is the concept of inductive definition, which characterizes objects like finite
lists, finite trees, finite games, finite strategy profiles, etc.  Similarly attached to
coinduction is the concept of coinductive definition which characterizes streams (infinite
lists), infinite trees, infinite games, infinite strategy profiles etc.  An inductive
definition yields the least set that satisfies the definition and a coinductive definition
yields the greatest set that satisfies the definition.  Associated with these definitions
we have inference principles.  For induction there is the famous \emph{induction
  principle} used in backward induction.
On coinductively defined sets of objects there is a principle like induction principle
which uses the fact that the set satisfies the definition (proofs by case or by pattern)
and that it is the largest set with this property.  Since coinductive definitions allow us
building infinite objects, one can imagine constructing a specific category of objects
with ``loops'', like the infinite word $(abc)^{`w}$ (\ie $abcabcabc...$) \index{infinite
  word} which is made by repeating the sequence $abc$ infinitely many times.\footnote{The
  notation $`a^{`w}$ for an infinite repetition of the word $`a$ is classical.}
  Other
examples with trees are given in Section~\ref{sec:l-bin-tree}, with infinite games and
strategy profiles in \Chapsec~\ref{chap:games}.  Such an object is a fixpoint, this means
that it contains an object like itself. For instance $(abc)^{`w} = abc(abc)^{`w}$ contains
itself. We say that such an object is defined as a cofixpoint.  To prove a property $P$ on
a cofixpoint $\ob=f(\ob)$, one assumes $P$ holds on $\ob$ (the $\ob$ in $f(\ob)$),
considered as a sub-object of $\ob$.  If one can prove $P$ on the whole object (on
$f(\ob)$), then one has proved that $P$ holds on~$\ob$.  This is called the
\emph{coinduction principle} a concept which comes from \citet{DBLP:conf/tcs/Park81},
\citet{DBLP:journals/tcs/MilnerT91}, and \cite{aczel88:_non_well_found_sets} and was
introduced in the framework we are considering by \citet{DBLP:conf/types/Coquand93}.
\citet{DBLP:journals/toplas/Sangiorgi09} gives a good survey with a complete historical
account.  To be sure to not be entangled, it is advisable to use a proof assistant which
implements coinduction, to build and to check the proof.  Indeed reasoning with
coinduction is sometimes so counter-intuitive that the use of a proof assistant is not
only advisable but compulsory.  For instance, we were, at first, convinced that in the
dollar auction the strategy profile consisting in both agents stopping at every step was a
Nash equilibrium, like in the finite case, and only failing in proving it mechanically
convinced us of the contrary and we were able to prove the opposite, namely that the
strategy profile ``stopping at every step'' is not a Nash equilibrium.  In the examples of
\Chapsec~\ref{chap:case_studies}, we have checked every statement using \Coq{} and in what
follows a sentence like ``we have proved that ...''  means that we have succeeded in
building a formal proof in \Coq.

\sssection{Backward coinduction and invariants}{Backward coinduction as a method for proving invariants}
\label{sec:backw-coinduct-as}

In infinite strategy profiles, the coinduction principles can be seen as follows:
a~property which holds on a strategy profile of an infinite extensive game is an
\emph{invariant}, \ie a property which is always true, along the histories and to prove
that this is an invariant one proceeds back to the past.  Therefore the name
\emph{backward coinduction} is appropriate, since it proceeds backward the histories, from
future to past.

\sssection{Backward induction vs backward coinduction}{Backward induction vs backward coinduction}
\label{sec:backw-induct-vs}

One may wonder the difference between the classical method, we call \emph{backward
  induction} and the new method we call \emph{backward coinduction}.  The main difference
is that backward induction starts the reasoning from the leaves, works only on finite
games and does not work on infinite games (or on finite strategy profiles), because it
requires a well-foundedness to work properly, whereas \emph{backward coinduction} works on
infinite games (or on infinite strategy profiles).  Coinduction is unavoidable on infinite
games, since the methods that consists in ``cutting the tail'' and extrapolating the
result from finite games or finite strategies profile to infinite games or infinite
strategy profiles cannot solve the problem or even approximate it.  It is indeed the same
erroneous reasoning as this of the predecessors of Weierstrass who concluded that since:
\index{Weierstrass}
\begin{displaymath}
`A p`:\nat,  f_p(x)=\sum_{n=0}^p b^n\cos(a^n x \pi),
\end{displaymath}
is differentiable everywhere then
\begin{displaymath}
f(x)=\sum_{n=0}^\infty b^n\cos(a^n x \pi).
\end{displaymath}
is differentiable everywhere whereas $f(x)$ is differentiable nowhere. 

Much earlier, during the IV$^{th}$ century BC, the improper use of inductive reasoning
allowed Parmenides and Zeno to negate motion and lead to \emph{Zeno's paradox of Achilles
  and the tortoise}.  This paradox was reported by Aristotle as follows: \index{Achilles
  and the tortoise}
\begin{quotation}
 \emph{``In a race, the quickest runner can never overtake the slowest, since the pursuer must first reach the point whence the pursued started, so
   that the slower must always hold a lead.''}

\hfill\emph{Aristotle}, Physics VI:9, 239b15
\end{quotation}
In Zeno's framework, Zeno's reasoning is correct, because by induction, one can prove that
Achilles will never overtake the tortoise.  Indeed this applies to the infinite sequence
of races described by Aristotle.  In each race of the sequence, the pursuer starts from
where the pursued started previously the race and the pursuer ends where the pursued started
in the current race.  By induction one can prove, but only for a sequence of races, the
truth of the statement ``Achilles will never overtake the tortoise''.  In each race
``Achilles does not overtake the tortoise''.  For the infinite race for which coinduction
would be needed, the result ``Achilles overtakes the tortoise'' holds.  By the way,
experience tells us that Achilles would overtake the tortoise in a real race and Zeno has
long been refuted by the real world.

\sssection{Von Neumann and coinduction}{Von Neumann and coinduction}
\label{sec:von-neumann}

As one knows, von Neumann
\citep{neumann28:_zur_theor_gesel,neumann44:_theor_games_econom_behav} is the creator of
game theory, whereas extensive games and equilibrium in non cooperative games are due to
\cite{Kuhn:ExtGamesInfo53} and \cite{Nash50}.  In the spirit of their creators all those
games are finite and backward induction is the basic principle for computing subgame
perfect equilibria \citep{selten65:_spiel_behan_eines_oligop_mit}.  This is not surprising
since \cite{neumann25:_axiom_mengen} is also at the origin of the role of well-foundedness
in set theory despite he left a door open for a not well-founded membership relation.  As
explained by \cite{DBLP:journals/toplas/Sangiorgi09}, research on anti-foundation
initiated by \cite{mirimanoff17:_les_antim_de_russel_et} are at the origin of coinduction
and were not well known until the work of \cite{aczel88:_non_well_found_sets}.

\sssection{Proof assistants vs automated theorem provers}{Proof assistants vs automated theorem provers}

\Coq{} is a proof assistant built by \emph{\citet{Coq:manual}},
see~\citet{BertotCasterant04} for a good introduction and notice that they call it
``interactive theorem provers'', which is a strict synonymous.  Despite both deal with
theorems and their proofs and are mechanized using a computer, proof assistants are not
automated theorem provers. In particular, they are much more expressive than automated
theorem provers and this is the reason why they are interactive.  For instance, there is
no automated theorem prover implementing coinduction.  Proof assistants are automated only
for elementary steps and interactive for the rest.  A specificity of a proof assistant is
that it builds a mathematical object called a (formal) proof which can be checked
independently, copied, stored and exchanged.  Following \cite{harrison-notices} and
\cite{dowek07:_les_metam_du_calcul}, we can consider that they are the tools of the
mathematicians of the XXI${^\textrm{th}}$ century.  Therefore using a proof assistant is a
highly mathematical modern activity.

The mathematical development presented here corresponds to a \Coq{} script\footnote{A \emph{script} is a list of commands of a proof assistant.} which can be found on the following url's:

\centerline{\url{http://perso.ens-lyon.fr/pierre.lescanne/COQ/Book/}}

\centerline{\url{http://perso.ens-lyon.fr/pierre.lescanne/COQ/Book/SCRIPTS/}}

\medskip

\sssection{Induction vs coinduction}{Induction vs coinduction}

To formalize structured finite objects, like finite games, one uses \emph{induction}, \ie
\begin{itemize}
\item a definition of basic objects
  \begin{itemize}
  \item in the case of natural numbers, induction provides an operator $0$ to build a natural
   number out of nothing,
 \item in the case of binary trees, a tree with non node, written $`()$,
 \item in the case of finite games induction provides an operator $\ogame\_\fgame$ to
   build a game out of nothing using an function that attribute payoffs to agents.
  \end{itemize}
and
\item a definition of the way to build new objects  
\begin{itemize}
  \item in the case of natural numbers, induction provides an operation \emph{successor}
   to build a natural number forma a natural number,
 \item in the case of binary trees, induction provides a binary operator which builds a
   tree with two trees,
  \item in the case of finite games induction provides an operator to build a game out of
    two subgames and a node.
  \end{itemize}
\end{itemize}
In the case of infinite objects like infinite games, one characterizes infinite objects
not by their construction, but by their behavior.  This characterization by
``observation'' is called \emph{coinduction}.  Coinduction is associated with the greatest
fixpoint.  The proof assistant \Coq{} offers a framework for coinductive definitions and
reasonings which are keys of our formalization.

\sssection{Acknowledgments}{Acknowledgments}

My research on game theory started during a visit at JAIST invited by Ren\'e Vestergaard,
then was continued by fruitful with St\'{e}phane Le~Roux and Franck Delaplace.  A decisive
step on infinite games and Nash Equilibrium was done by Matthieu Perrinel
\citep{LescannePerrinel}.  I thank all them.

\chapitre{The concepts through examples}{The concepts through examples}
\label{cha:conc-thro-exampl}

We think that examples are the best way to present concepts. In this \chaptersection{} we present a
simple example of an infinite game useful in what follows and two examples of structures
meant to introduce smoothly induction.


\ssection{A paradigmatic example: the  ${0, 1}$~game}{A paradigmatic example: the  ${0, 1}$~game}
\label{sec:01}\index{$0,1$ game}

Classically, an extensive game is considered a labelled oriented tree, in which both nodes
and arcs are labelled.  In other words, there is a set of nodes and a set of arcs.  An arc
connects a node to another node in such a way that there is no circuit in the graph, \ie
no path which goes from a node to itself when following the arcs.  An \emph{internal node}
is a node which is connected to another node and an \emph{external node} or a \emph{leaf}
is a node which is not connected to another.  Internal nodes are labelled by names of
players and represents a turn in the game and the label of the internal node tells us
which player has the turn.  External nodes represent the end of the game and are labelled
by the function that assigns a payoff\footnote{or a cost in some cases.} to each player.
Arcs are labeled by choices, more precisely choices made by the player who has the turn
and show how the choice made by the player who has the turn leads in another position in
the game.  In our formalization we assume that one can go to a position in which the same
player has the turn, like

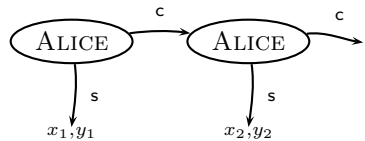
\[\begin{psmatrix}[colsep=20pt,rowsep=20pt] 
  & {\ovalnode{aa}{\scriptstyle \Alice}} &{\ovalnode{ab}{\scriptstyle \Alice}}
      &[name=ac]\\
      &[name=aad] {\scriptstyle x_1,y_1} &[name=abd] {\scriptstyle x_2,y_2} %
      & 
      \ncarc[arrows=->]{aa}{ab}\Aput{\scriptstyle \mathsf{c}} %
      \ncarc[arrows=->]{ab}{ac} \Aput{\scriptstyle \mathsf{c}} %
     \ncarc[arrows=->]{aa}{aad} \Aput{\scriptstyle \mathsf{s}} %
      \ncarc[arrows=->]{ab}{abd} \Aput{\scriptstyle \mathsf{s}} %
\end{psmatrix}
\]
but this situation never occurs in examples we consider.  In this \booksecvirg we propose a
presentation of game less descriptive and more structural in the line of what is done in
computer science when describing infinite computations. 

\figzeroone

To illustrate this description we propose a running example of an infinite extensive game
which we call the \emph{$0, 1$ game} (which formal name is \zone{}) because the only
utilities (or payoffs) are $0$ and $1$.  The game is infinite and is represented by a kind
of infinite backbone (see Figure~\ref{fig:zone}) in which each internal node is connected
to a leaf and to another internal node.  We assume that there are two players, namely
\Alice{} and \Bob, hence two labels on the internal nodes.  We also assume that players
play one after the other.  They have two choices, \textsf{s}~for \textsf{stop} and
\textsf{c} for \textsf{continue}.  The arc labeled \textsf{s} is connected to a leaf and the arc
labeled \textsf{c} is connected to another internal node.  The leaves have labels that are
payoff functions.  The leaf connected with an internal node labeled with \Alice{} is
labeled with function $\{\Alice "|->" 0, \Bob "|->" 1\}$ (meaning $\Alice$'s payoff is $0$
and $\Bob$'s payoff is~$1$) whereas the leaf connected with an internal node labeled with
$\Bob$ is labeled with function $\{\Alice "|->" 1, \Bob "|->" 0\}$ (meaning the payoffs
are reversed).  The \emph{$0, 1$ game} is a specific case of a binary game which are
presented using the formalism for defining infinite objects coinductively.  Binary games
have two kinds of labels called $\lft$ and $\rgt$.  In the \emph{0, 1} game $\lft$ stands for
\textsf{c} (continues) and $\rgt$ stands for \textsf{s} (stops).  In what follows the $0, 1$
game will be formally defined as a coinductive structure defined by a fixpoint, that is

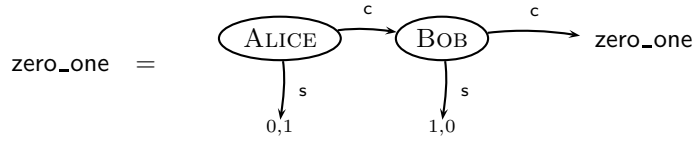
\[
\raisebox{20pt}{\zone \quad= }
\begin{psmatrix}[colsep=20pt,rowsep=20pt] 
  & {\ovalnode{aa}{\scriptstyle \Alice}} &{\ovalnode{ab}{\scriptstyle \Bob}}
      &\rnode{ac}{\qquad \zone}\\
      &[name=aad] {\scriptstyle 0,1} &[name=abd] {\scriptstyle 1,0} %
      & 
      \ncarc[arrows=->]{aa}{ab}\Aput{\scriptstyle \mathsf{c}} %
      \ncarc[arrows=->]{ab}{ac} \Aput{\scriptstyle \mathsf{c}} %
     \ncarc[arrows=->]{aa}{aad} \Aput{\scriptstyle \mathsf{s}} %
      \ncarc[arrows=->]{ab}{abd} \Aput{\scriptstyle \mathsf{s}} %
\end{psmatrix}
\]

\ssection{Coinduction}{Coinduction through examples}
\label{sec:inf_obj}

We now leave the descriptive approach for the structural approach.
We introduce the concept of coinduction through two examples: histories of
sequential games, \ie the sequences of choices performed by agents along the run of a
game according to a strategy profile, and binary trees, trees in which there are two
subtrees at each node.

\sssection{Histories}{Histories}

Infinite objects have peculiar behaviors. To start with a simple example, let us have a
look at \emph{histories} in games \citep[Chap. 5]{osborne04a}. In a game, agents make
\emph{choices}.  In an infinite game, agents can make finitely many choices before ending,
if they reach a terminal node, or infinitely many choices, if they run forever. Choices
are recorded in a \emph{history} in both cases.  A history is therefore a finite or an
infinite list of choices.  In this \booksecvirg we consider that there are two possible
choices: \lft~and \textsf{r} (\lft{} for ``left'' and \textsf{r} for ``right'').  Since a
history is a potentially infinite object, it cannot be defined by structural
induction.\footnote{In type theory, a type of objects defined by induction is called an
  \coqdockw{Inductive}, a shorthand sWe can for \emph{inductive type}.}  On the contrary, the
type\footnote{Since we are in type theory, the basic concept is \emph{``type''}. Since we
  are using only a small part of type theory, it would not hurt to assimilate naive types
  with naive sets.} \emph{History} has to be defined as a \coqdockw{CoInductive}, \ie
\index{coinduction}by coinduction, that is a mechanism which defines infinite objects and
allows to reason on them.  Let us use the symbol~$[~]$ for the empty history and
the binary operator~$::$ for non empty histories.  When we write $a::h$ we mean that the
history starts with $a$ and follows with the history $h$.  For instance, the finite
history $\lft\rgt\lft$ can be written $\lft::(\rgt::(\lft::[\ ]))$.  If $h$ is the
history $\lft^{`w}$ (an infinite sequence of $\lft$'s) $\lft:: h$ or $\lft::\lft^{`w}$ is
the history that starts with $\lft$ and follows with infinitely many $\lft$'s.  The reader
recognizes that $\lft::\lft^{`w}$ is $\lft^{`w}$ itself.
To define histories
coinductively we say the following:
\begin{quotation}
  A \textbf{coinductive} history (or a finite or infinite history) is
  \begin{itemize}
  \item either the empty history $[\ ]$,
  \item or a history of the form $a::h$, where $a$ is a choice and $h$ is  a history.
  \end{itemize}
\end{quotation}
The word ``coinductive'' says that we are talking about finite or infinite objects.  This
should not be mixed up with finite histories which will be defined inductively as
follows:\pagebreak[3]
\begin{quotation}
  An \textbf{inductive} history (or a finite history) is built as
  \begin{itemize}
  \item either the empty history $[\ ]$,
  \item or a finite non empty history which is the composition of a choice $a$ with a
    finite history $h_f$ to make the finite history ${a:: h_f}$.
  \end{itemize}
\end{quotation}
Notice the use of the participial ``built'', since in the case of induction, we say how
objects are built, because they are built finitely.  The $0,1$ game has the family of
histories $\mathsf{c}^*\mathsf{s} \cup \mathsf{c}^{`w}$, meaning that a history is either a sequence of $\mathsf{c}$'s followed
by a \textsf{s}, or an infinite sequence of $\mathsf{c}$'s.   Consider now an arbitrary
infinite binary game\footnote{An game which does not have a centipede structure, \ie which
  does not have a backbone.} with histories made of \lft's and \rgt's.
Let us now consider four families of
histories:

\medskip

\begin{tabular}[h]{|l|l|}
  \hline
  $\H_0$ & {\scriptsize The family of finite histories}\\
  $\H_1$ & {\scriptsize The family of finite histories or of histories which end with an infinite sequence of \lft's}\\
  $\H_2$ & {\scriptsize The family of finite histories or infinite histories which contain infinitely many \lft's}\\
  $\H_{\infty}$ & {\scriptsize The family of finite or infinite histories}\\
  \hline
\end{tabular}

\medskip

We notice that $\H_0\subset\H_1\subset\H_2\subset\H_{\infty}$.  If $\H$ is a set of histories, we write $\lft::\H$ the set $\{h \in \H_{\infty} \mid \exists h' \in \H, h = c :: h'\}$.
We notice that $\H_0$, $\H_1$, $\H_2$ and $\H_{\infty}$ are solutions of the fixpoint equation :
\[\H = \{[~]\} \ \cup\ \lft::\H \ \cup\ \textsf{r}::\H.\]
in other words
\begin{eqnarray*}
  \H_0 &=& \{[~]\} \ \cup\ \lft::\H_0 \ \cup\ \textsf{r}::\H_0\\
  \H_1 &=& \{[~]\} \ \cup\ \lft::\H_1 \ \cup\ \textsf{r}::\H_1\\
  \H_2 &=& \{[~]\} \ \cup\ \lft::\H_2 \ \cup\ \textsf{r}::\H_2\\
  \H_{\infty} &=& \{[~]\} \ \cup\ \lft::\H_\infty \ \cup\ \textsf{r}::\H_{\infty}
\end{eqnarray*}
Among all the fixpoints of the above equation, $\H_0$ is the least fixpoint and describes the inductive type associated with this equation, that is the type of the finite histories and
$\H_\infty$ is the greatest fixpoint and describes the coinductive type associated with this equation, that is the type of the infinite and infinite histories.  The principle that says
that given an equation, the least fixpoint is the inductive type associated with this equation and the greatest fixpoint is the coinductive type associated with this equation is very
general and will be used all along this \bookarticlepoint 

The \Coq{} vernacular,  is more verbose, but also more precise in describing the
\textsf{CoInductive} type \emph{History}, (see Appendix~\ref{sec:excerpts-coq-devel} for a precise
definition).   The word
\textbf{coinductive} guarantees that we define actually infinite objects and attach to the
objects of type \emph{History} a specific form of reasoning, called \emph{coinduction}.
In coinduction, we assume that we ``know'' an infinite object by observing it through its
definition, which is done by a kind of peeling.  Since on infinite objects there is no
concept of being smaller, one does not reason by saying ``I~know that the property holds
on smaller objects let us prove it on the object''.  On the contrary one says ``Let us
prove a property on an infinite object.  For that peel the object, assume that the property
holds on the peeled object and prove that it holds on the whole object''.  One does not
say that the object is smaller, just that the property holds on the peeled object.  The
above presentation is completely informal, but it has been, formally founded by Christine
Paulin in the theory
of \Coq, after the pioneer works of \cite{DBLP:conf/tcs/Park81} and \cite{Milner89}, using
the concept of greatest fixpoint in type theory \citep{DBLP:conf/types/Coquand93} (see
\cite{DBLP:journals/toplas/Sangiorgi09} for a survey).  \cite{BertotCasterant04} present
the concepts in \Chapsec~13 of their book.

\ssssection{Bisimilarity}

By just observing them, one cannot prove that two objects which have exactly the same
behavior are equal, we can just say that they are observably equivalent.  Observable
equivalence is a relation weaker than equality\footnote{We are talking here about
  \emph{Leibniz equality}, not about \emph{extensional equality} see
  appendix~\ref{sec:eq}.}, called \emph{bisimilarity} and defined on \emph{History} as a
\textsf{CoInductive} (see Appendix~\ref{sec:excerpts-coq-devel} for a fully formal
definition in the \Coq{} vernacular):

\begin{quotation}
  Bisimilarity $\hbis$ on \emph{histories} is defined \textbf{coinductively} as follows:
  \begin{itemize}
  \item $[\ ] \hbis [\ ]$,
  \item $h \hbis h'$ implies $\forall a:Agent, a :: h \hbis a::h'$.
  \end{itemize}
\end{quotation}

This means that two histories are bisimilar if either both are null or for composed
histories, if both have the same head and the rests of both histories are
bisimilar.\footnote{Bisimilarity is related with \emph{p-morphisms} and \emph{zigzag
    relations} in modal logic. See \citep{DBLP:journals/toplas/Sangiorgi09} for a survey.}
One can prove that two objects that are equal are bisimilar, but not the other way around,
because for two objects to be equivalent by observation, does not mean that they have the
same identity. To illustrate the difference between bisimilarity and equality of infinite
objects let us consider for example two infinite histories $`a_0$ and $`b_0$ that are
obtained as solutions of two equations.  Let $`a_n = c(n):: `a_{n+1}$, where $c(n)$ is
$(\textbf{if}~even(n)~\textbf{then}~\lft~\textbf{else}~\textsf{r})$, and $`b_p$ is $\lft
:: \textsf{r} :: `b_{p+1}$.  We know that if we ask for the $5^{th}$ element of $`a_0$ and
$`b_0$ we will get $\lft$ in both cases, and the $2p^{th}$ element will be $\textsf{r}$ in
both cases, but we have no way to prove that $`a_0$ and $`b_0$ are equal, i.e., have
exactly the same structure.  Actually the picture in Figure~\ref{fig:a0b0} shows that they
look different and there is no hope to prove by induction, for instance, that they are the
same, since they are not well-founded.  We see that in the first history, for all $p$, we
have $`a_0 = `a_{2p+1}$ and, we can see $`a_0$ as fixpoint of the system of equations:
\[x_{`a} = \lft :: x_{`a'} \qquad \qquad x_{`a'} = \rgt :: x_{`a}\] and for the second
history, for all $p$ we have also $`b_0= `b_{2p+1}$ and we can see $`b_0$ as the fixpoint of the equation
\[y_{`b} = \lft :: \rgt :: y_{`b}.\]

\begin{figure}[!tbh]
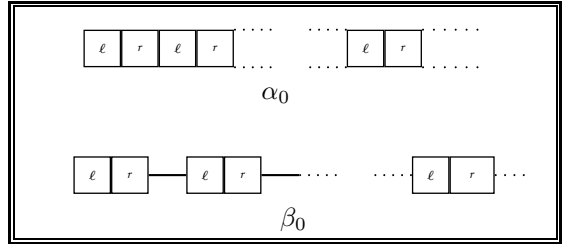

  \centering
  \doublebox{
  \qquad\parbox{.5\textwidth}{
    \vspace*{.7cm}
    \psframe[linewidth=.5pt,fillstyle=solid](0,0)(.5,.5)
    \psframe[linewidth=.5pt,fillstyle=solid](.5,0)(1,.5) %
    \psframe[linewidth=.5pt,fillstyle=solid](1,0)(1.5,.5) %
    \psframe[linewidth=.5pt,fillstyle=solid](1.5,0)(2,.5) %
    \rput(.25,.25){\tiny \lft} %
    \rput(.75,.25){\tiny \rgt} %
    \rput(1.25,.25){\tiny \lft} %
    \rput(1.75,.25){\tiny \rgt} %
    \psline[linestyle=dotted](2,.5)(2.5,.5) %
    \psline[linestyle=dotted](2,0)(2.5,0) %
    \psline[linestyle=dotted](3,.5)(3.5,.5) \psline[linestyle=dotted](3,0)(3.5,0) %
    \psframe[linewidth=.5pt,fillstyle=solid](3.5,0)(4,.5)
    \psframe[linewidth=.5pt,fillstyle=solid](4,0)(4.5,.5) %
    \rput(3.75,.25){\tiny \lft} \rput(4.25,.25){\tiny \rgt}
    \psline[linestyle=dotted](4.5,.5)(5.25,.5)
    \psline[linestyle=dotted](4.5,0)(5.25,0) %
    \\\hspace*{2.25cm} $`a_0$

    \bigskip\bigskip

    \hspace*{-2.25cm}
    \psframe[linewidth=.5pt,fillstyle=solid](2,0)(2.5,.5)%
    \psframe[linewidth=.5pt,fillstyle=solid](2.5,0)(3,.5) \rput(2.25,.25){\tiny \lft}
    \rput(2.75,.25){\tiny \rgt} \psline(3,.25)(3.5,.25) %
    \psframe[linewidth=.5pt,fillstyle=solid](3.5,0)(4,.5)%
    \psframe[linewidth=.5pt,fillstyle=solid](4,0)(4.5,.5) \rput(3.75,.25){\tiny \lft}
    \rput(4.25,.25){\tiny \rgt} \psline(4.5,.25)(5,.25) %
     \psline[linestyle=dotted](5,.25)(5.5,.25)
    \psline[linestyle=dotted](6,.25)(6.5,.25)
    \psframe[linewidth=.5pt,fillstyle=solid](6.5,0)(7,.5)%
    \psframe[linewidth=.5pt,fillstyle=solid](7,0)(7.6,.5) \rput(6.75,.25){\tiny \lft}
    \rput(7.3,.25){\tiny \rgt} \psline[linestyle=dotted](7.6,.25)(8,.25)
 \\\hspace*{2.5cm} $`b_0$
  }
}
  \caption{The picture of two bisimilar histories}
  \label{fig:a0b0}
\end{figure}

As we said, we do not consider equality among infinite objects, but only bisimilarity.
Why?  The reason is that with the kind of reasoning we use, we can only prove that two
objects ar bisimilar, not that they are equal.\footnote{We can consider a definition of infinite
object where objects are equal if they have the same elements, but such a set of objects is
obtained by quotient of the set of histories  by the bisimilarity relation.   This way, we
loose the structure of the infinite objects as we described them. For us it is important to
keep the structure of the objects.}

\ssssection{Always}

A property $P$ can be true on an infinite history.  For instance \emph{``there exists
  an~\lft{} in the sequence''}.  But we can also say that a property of a history is
always true, that is true for all the sub-histories of the history.  For instance,
\emph{``there is always  an \lft{} further in the sequence''}, this also means
\emph{``there exists infinitely many \lft's in the history''}.  The operator which
transforms a property~$P$ in a property \emph{always}~$P$ is called a \emph{modality}.
The modality \emph{always} is written~$\Box$ and we write $\Box P$ instead of
\emph{always}~$P$.

\sssection{Binary  trees}{Infinite or finite binary  trees}
\label{sec:l-bin-tree} \index{binary tree} \index{infinite binary tree}

As an example of a coinductive definition consider \emph{binary tree}, \ie the type of finite and infinite binary
trees.

\begin{figure}[thb]
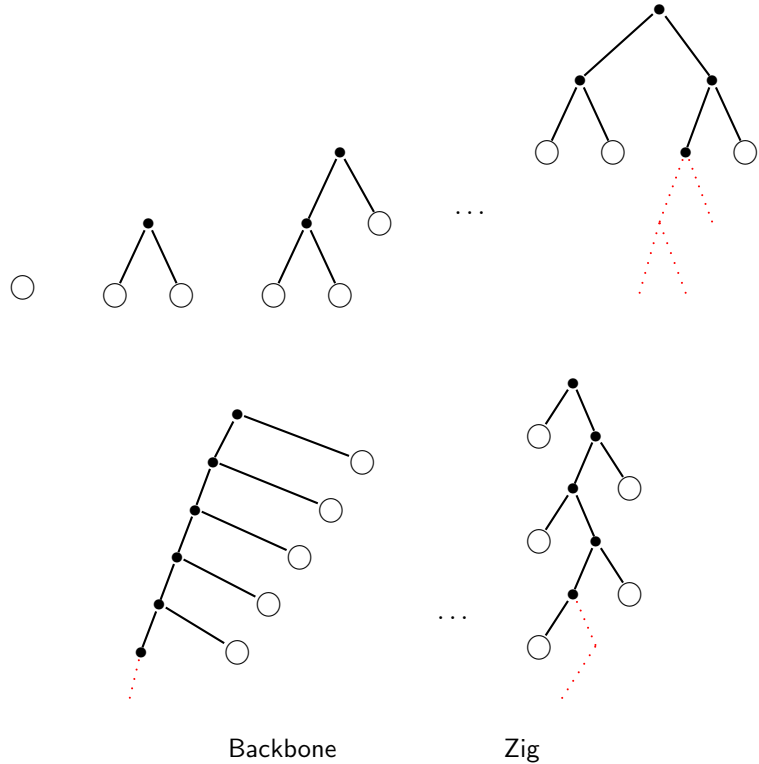

  \begin{center}
    $`()$\qquad
    \begin{psmatrix}[rowsep=15pt,colsep=5pt]
      &&[name=a]{$\bullet$} & \\
      &[name=b]{$`()$} && [name=c]{$`()$} %
      \ncline{a}{b} \ncline{a}{c}
    \end{psmatrix}\qquad
    \begin{psmatrix}[rowsep=15pt,colsep=5pt]
      &&&[name=a]{$\bullet$} & \\
      &&[name=b]{$\bullet$} && [name=c]{$`()$} \\%
      &[name=d]{$`()$}&&[name=e]{$`()$} 
      \ncline{a}{b} \ncline{a}{c} \ncline{b}{d} \ncline{b}{e}
    \end{psmatrix}
    \qquad \raisebox{30pt}{$\ldots$}
    \begin{psmatrix}[rowsep=15pt,colsep=5pt]
      &&&&&&&[name=a]{$\bullet$} \\
      &&&&[name=b]{$\bullet$} &&&&& [name=c]{$\bullet$} \\%
      &&&[name=d]{$`()$}&&[name=e]{$`()$}
      &&&[name=f]{$\bullet$}&&[name=g]{$`()$} \\
      &&&&&&&[name=h]&&[name=i] \\ %
      &&&&&&[name=j]&&[name=k] %
      \ncline{a}{b} \ncline{a}{c} %
      \ncline{b}{d} \ncline{b}{e} \ncline{c}{f} \ncline{c}{g} %
      \rouge{\ncline[linestyle=dotted,linecolor=red]{f}{h}} \ncline[linestyle=dotted,linecolor=red]{f}{i}
      \rouge{\ncline[linestyle=dotted,linecolor=red]{h}{j}} \ncline[linestyle=dotted,linecolor=red]{h}{k}
    \end{psmatrix}
  \end{center}

  \medskip

  \begin{center}
    \begin{psmatrix}[rowsep=6pt,colsep=1.8pt]
      &&&&&&&[name=a]{$\bullet$} \\
      &&&&&&[name=b]{$\bullet$} &&&&& [name=c]{$`()$} \\%
      &&&&&[name=d]{$\bullet$}&&&&&[name=e]{$`()$} \\
      &&&&[name=f]{$\bullet$}&&&&&[name=g]{$`()$} \\
      &&&[name=h]{$\bullet$}&&&&&[name=i]{$`()$} \\ %
      &&[name=j]{$\bullet$}&&&&&[name=k]{$`()$}\\ %
      &[name=l] %
      \ncline{a}{b} \ncline{a}{c} %
      \ncline{b}{d} \ncline{b}{e} %
      \ncline{d}{f} \ncline{d}{g} %
      \ncline{f}{h} \ncline{f}{i} %
      \ncline{h}{j} \ncline{h}{k} %
      \ncline[linestyle=dotted,linecolor=red]{j}{l}
    \end{psmatrix}
    \qquad\raisebox{30pt}{$\ldots$}\qquad
    \begin{psmatrix}[rowsep=8pt,colsep=1.8pt]
      &&&&[name=a]{$\bullet$} \\
      &[name=b]{$`()$} &&&&& [name=c]{$\bullet$} \\%
      &&&&[name=d]{$\bullet$}&&&&&[name=e]{$`()$} \\
      &[name=f]{$`()$}&&&&&[name=g]{$\bullet$} \\
      &&&&[name=h]{$\bullet$}&&&&&[name=i]{$`()$} \\ %
      &[name=j]{$`()$}&&&&&[name=k]\\ %
      &&&[name=l] \ncline{a}{b} \ncline{a}{c} %
      \ncline{c}{d} \ncline{c}{e} \ncline{d}{f} \ncline{d}{g} \ncline{g}{h} \ncline{g}{i} \ncline{h}{j} \ncline[linestyle=dotted,linecolor=red]{h}{k}
      \ncline[linestyle=dotted,linecolor=red]{k}{l}
    \end{psmatrix}
  \end{center}

  \begin{center}
    \textsf{Backbone} \qquad \qquad \qquad \textsf{Zig}
  \end{center}
  \caption{Coinductive binary trees}
  \label{fig:lazy}
\end{figure}
\medskip

\begin{quotation}
  A \textsf{coinductive} \emph{binary tree} (or a lazy binary tree or a finite-infinite binary tree) is
 \ifElsev \begin{itemize}[$\star$] \else \begin{itemize} \fi
  \item either the empty binary tree $`()$,
  \item or a binary tree of the form $t `. t'$, where $t$ and $t'$ are binary trees.
  \end{itemize}
\end{quotation}

By the keyword \textsf{coinductive} we mean that we define a coinductive set of objects, hence we accept infinite objects.  Some coinductive binary trees are
given on Figure~\ref{fig:lazy}.  We define on coinductive binary trees a \emph{predicate} which has also a coinductive definition:

\begin{quotation}
  A binary tree is \emph{infinite} if (coinductively)
\ifElsev \begin{itemize}[$\star$] \else \begin{itemize} \fi 
  \item either its left subtree is \emph{infinite}
  \item or its right subtree is \emph{infinite}.
  \end{itemize}
\end{quotation}
Can we speak about a specific infinite tree?  Yes provided we can define it.  This can be done
as a fixpoint, actually a cofixpoint since we speak about an infinite object.  Let
us define an infinite binary tree with an infinite path that goes left, then right, then
left, then right, then left, forever (Figure~\ref{fig:lazy}).  We call \textsf{zig} this
infinite tree.  Its definition goes with another infinite tree called \textsf{zag}.

We define two trees that we call \textsf{zig} and \textsf{zag}.
\begin{quotation}
  \textsf{zig} and \textsf{zag} are defined together as \emph{cofixpoint}s as follows:
\ifElsev \begin{itemize}[$\star$] \else \begin{itemize} \fi
  \item \textsf{zig} has  $`()$ as left subtree and \textsf{zag} as right subtree,
  \item \textsf{zag} has \textsf{zig} as left subtree and $`()$  as right subtree.
  \end{itemize}
\end{quotation}
This says that \textsf{zig} and \textsf{zag} are the greatest solutions\footnote{In this case, the least solutions are uninteresting as they are
  objects nowhere defined. Indeed there is no basic case in the inductive definition.} of the two simultaneous equations:
\begin{eqnarray*}
  \mathsf{zig} &=& `()  `.  \mathsf{zag}\\
  \mathsf{zag} &=& \mathsf{zig}  `. `()
\end{eqnarray*}

\ziginf{}

It is common sense that \textsf{zig} and \textsf{zag} are infinite, but to prove that
\emph{``\textsf{zig} is infinite''} using the \textsf{cofix} tactic\footnote{The
  \textsf{cofix} tactic is a method due to Christine Paulin and proposed by the proof
  assistant \Coq{} which implements coinduction on cofixpoint objects.  Roughly speaking,
  it attempts to prove that a property is an \emph{invariant}, by proving it is preserved
  along the infinite object.  Here \emph{`` is infinite'' } is such an invariant on
  \textsf{zig}.}, we do as follows: assume \emph{``\textsf{zig} is infinite''}, then
\textsf{zag} is infinite, from which we get that \emph{``\textsf{zig} is infinite''}.
Since we use the assumption on a strict subtree of \textsf{zig} (the direct subtree of
\textsf{zag}, which is itself a direct subtree of \textsf{zig}) we can conclude that the
\textsf{cofix} tactic has been used properly and that the property holds, namely that
\emph{``\textsf{zig} is infinite''}.  We have proved that \emph{``\textsf{zig} is
  infinite'''} is an invariant along the infinite binary trees \textsf{zig} and
\textsf{zag}.  The cofix reasoning is pictured on Fig.\ref{fig:ziginf}, where the square
box represents the predicate \emph{is infinite}.  Above the rule, there is the step of
coinduction and below the rule the conclusion, namely that the whole \textsf{zig} is
infinite.  We let the reader prove that \emph{``\textsf{backbone} is infinite''}, where
\textsf{backbone} is the greatest fixpoint of the equation:
\begin{center}
  \textsf{backbone} \ = \ \textsf{backbone}  $`.\ `()$
\end{center}

Interested readers may have a look
at~\citet*{coupet-grimal03:_axiom_of_linear_temp_logic,DBLP:journals/fac/Coupet-GrimalJ04,DBLP:journals/corr/abs-0904-3528,bertot05:_filter_coind_stream_applic_to_eratos_sieve,bertot07:_affin_funct_and_series_with}
  and especially \citet[chap.~13]{BertotCasterant04} for other examples of \textsf{cofix} reasoning.
\chapitre{Games and strategy profiles}{Games and strategy profiles}
\label{chap:games}

We start with a formal and inductive presentation of finite games, which is extended in
the next section to a description of infinite games.  The \chaptersection{} ends with a presentation
of equilibria in infinite games: Nash equilibria and subgame perfect equilibria. 

In classical textbooks, finite and infinite games are presented through their
histories. But in the framework of a proof assistant or just to make rigorous proofs, it
makes sense to present them structurally.  Therefore, games are rather naturally seen as
either a leaf to which a \emph{utility function} (a function that assigns a utility, a
payoff or a cost to
each agent, aka an \emph{outcome}) is attached or a node which is associated to an agent
and two subgames.  If agents are \Alice{} and \Bob{} and utilities are natural
numbers, a utility function can be the function $\Alice{}"|->" 3, \Bob{}
"|->" 2$. 

\ssection{Finite Games}{Finite Games}
\label{sec:finite-games}

We restrict to infinite games in which each player has two choices at each turn.  Such
finite extensive game can be seen as built by putting together a player $a$ and two games
$g_{\lft}$ and $g_{\rgt}$, which correspond to either choice made by the
player.  We write $\ogame a, g_{\lft}, g_{\rgt}\fgame$ this game.  We need also a base
case, which is actually what is seen usually as the ``end'' of a game and which is used here as
the basis  which every finite game is based upon.    Actually it is a degenerated game where
players do not play but just receive their payoffs.  Assume there are two players $\Alice$
and $\Bob$ and $p_A$ is the payoff for $\Alice$ and $p_B$ is the payoff for $\Bob$.  This
is the utility fonction $f \ `=\ \Alice "|->" p_A, \Bob "|->" p_B$.  We write $\ogame f\fgame$ this
kind of game.  A finite binary game is a game obtained by
applying repeatedly applications of $\ogame a, \_, \_ \fgame$ to games of the form
$\ogame f \fgame$.  In other words:

\begin{quotation}
  The type \emph{Finite Game} \index{finite game} is defined as an \textbf{inductive} as follows:
  \begin{itemize}
  \item a \emph{Utility function} makes a \emph{Finite Game},
  \item an \emph{Agent} and two \emph{Finite Game}s make a \emph{Finite Game}.
  \end{itemize}
\end{quotation}

 Hence one builds a finite game in two ways: either a given utility function~$f$ is
  encapsulated to make the game \ogame\coqdocid{f}\fgame, or an agent \coqdocid{a} and two
  games $g_l$ and $g_r$ are given to make the game \ogame\coqdocid{a}, $g_l$, $g_r$\fgame.
  Notice that in such games, it can be the case that the same agent \coqdocid{a} has the
  turn twice in a row, like in the game ${\ogame\coqdocid{a}, \ogame
  \coqdocid{a}, g_1,g_2\fgame, g_3\fgame}$.

\ssection{Infinite Games}{Infinite Games}
\label{sec:infinite-games}

 We study games that ``can'' be infinite and ``can'' have finite
or infinite branches, like the $0,1$ game.

\begin{quotation}
  The type \emph{Game} \index{game} is defined as a \textbf{coinductive} as follows:
  \begin{itemize}
  \item a \emph{Utility function} makes a \emph{Game},
  \item an \emph{Agent} and two \emph{Game}s make a \emph{Game}.
  \end{itemize}
\end{quotation}

A \emph{Game} is either a leaf (a terminal node) or a composed game made of an agent (the
agent who has the turn) and two subgames (the formal definition in the \Coq{} vernacular
is given in the appendix~\ref{sec:excerpts-coq-devel}).  Like for finite games, we use the expression
$\ogame f \fgame$ to denote the leaf game associated with the utility function $f$ and
the expression $\ogame a,g_l,g_r\fgame$ to denote the game with agent $a$ at the
root and two subgames $g_l$ and~$g_r$.  For instance, the game we would
draw:
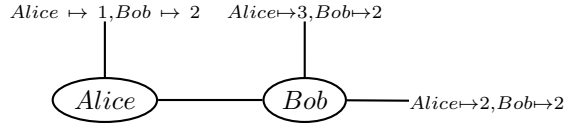
\label{pag:game}

\[\begin{psmatrix}[colsep=10pt,rowsep=20pt]
[name=a]{\scriptstyle Alice ~\mapsto~ 1, Bob ~\mapsto~ 2}&
[name=b]{\scriptstyle Alice \mapsto 3, Bob \mapsto 2} \\
{\ovalnode{c}{Alice}}  & {\ovalnode{d}{Bob}} &
[name=e]{\scriptstyle Alice \mapsto 2, Bob \mapsto 2}
 \ncline{c}{a}
  \ncline{c}{d}
  \ncline{d}{b}
  \ncline{d}{e}
\end{psmatrix}
\]

\medskip

\noindent is represented by the term:

\medskip

  \ogame \Alice{}, \ogame$\scriptstyle Alice ~\mapsto~ 1, Bob, ~\mapsto~2$\fgame, \ogame\Bob{} 
  \ogame{\scriptsize $Alice \mapsto 2, Bob \mapsto 2$}\fgame,
  \ogame{\scriptsize$Alice \mapsto 3, Bob \mapsto 2$}\fgame\fgame

\medskip

Concerning comparisons of utilities we consider a very general setting where a utility is
no more that a type (a ``set'') with a preference which is a preorder, i.~e., a transitive
and reflexive relation, and which we write \leut. A preorder is enough for what we want to
prove.  By using a very general preorder, it makes extremely easy to go from payoff to
cost, we have just to switch the direction of $\leut$ keeping the same carrier.
We assign to the leaves, a utility function which associates a utility to each agent.

Like for histories, to describe an infinite game one uses a fixpoint equation. For
instance to describe the $0,1$~game one uses the equation:

\medskip

\noindent
\begin{small}
  \begin{math}
    \zone\, = \, \ogame~\Alice,~ \ogame ~\Bob,~ \zone,~\ogame {\scriptstyle\Alice ~"|->"~
      0, \Bob ~"|->"~1}\fgame\fgame, \ogame {\scriptstyle \Alice ~"|->"~ 1, \Bob
      ~"|->"~0}\fgame\fgame
  \end{math}
\end{small}

\ssection{Infinite Strategy Profiles}{Infinite Strategy Profiles}
\label{sec:inf_stra}\index{strategy profile}

The main concept of this \bookarticle{} is this of infinite strategy profile which is a coinductive.
More specifically, in this \booksecvirg we focus on infinite binary strategy profiles associated
with infinite binary games.
\begin{quotation}
  The type of \emph{Strategy Profiles} is defined as a \textbf{coinductive} as follows:
  \begin{itemize}
  \item a \emph{Utility function} makes a \emph{Strategy Profile}.
  \item an \emph{Agent}, a \emph{Choice} and two \emph{Strategy Profiles} make a \emph{Strategy Profile}.
  \end{itemize}
\end{quotation}

Basically\footnote{The formal definition in the \Coq{} vernacular is given in
  appendix~\ref{sec:excerpts-coq-devel}.} an infinite strategy profile which is not a leaf
is a node with four items: an agent, a choice, two infinite strategy profiles.  A strategy
profile is the same as a game, except that there is a choice.  In what follows, since we
consider equilibria, we only address strategy profiles.  Strategy profiles of the first
kind are written $\og f \fg$ and strategy profiles of the second kind are written $\og
a,c,s_l,s_r\fg$.  In other words, if between the ``$\og$'' and the ``$\fg$'' there is one
component, this component is a utility function and the result is a leaf strategy profile
and if there are four components, this is a node strategy profile.  For instance, with the
game of page~\pageref{pag:game} one can associate at least the following strategy
profiles:

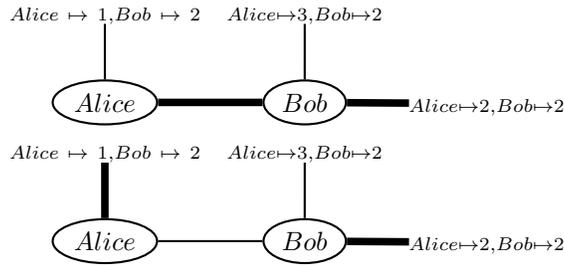
\[\begin{psmatrix}[colsep=10pt,rowsep=20pt]
[name=a]{\scriptstyle Alice ~\mapsto~ 1, Bob ~\mapsto~ 2}&
[name=b]{\scriptstyle Alice \mapsto 3, Bob \mapsto 2} \\
{\ovalnode{c}{Alice}}  & {\ovalnode{d}{Bob}} &
[name=e]{\scriptstyle Alice \mapsto 2, Bob \mapsto 2}
 \ncline{c}{a}
  \ncline[linewidth=.1]{c}{d}
  \ncline{d}{b}
  \ncline[linewidth=.1]{d}{e}
\end{psmatrix}
\]
\[\begin{psmatrix}[colsep=10pt,rowsep=20pt]
[name=a]{\scriptstyle Alice ~\mapsto~ 1, Bob ~\mapsto~ 2}&
[name=b]{\scriptstyle Alice \mapsto 3, Bob \mapsto 2} \\
{\ovalnode{c}{Alice}}  & {\ovalnode{d}{Bob}} &
[name=e]{\scriptstyle Alice \mapsto 2, Bob \mapsto 2}
 \ncline[linewidth=.1]{c}{a}
  \ncline{c}{d}
  \ncline{d}{b}
  \ncline[linewidth=.1]{d}{e}
\end{psmatrix}
\]
\medskip
\noindent which correspond to the expressions

  $\og$\Alice{},\rgt,$\og$ $\scriptstyle Alice ~\mapsto~ 1, Bob ~\mapsto~ 2$ $\fg$,\\
  \hspace*{65pt} $\og$\Bob{},\lft,$\og${\scriptsize $Alice \mapsto 2, Bob \mapsto 2$}$\fg$,$\og${\scriptsize$Alice \mapsto 3, Bob \mapsto
    2$}$\fg$ $\fg$ 

and

  $\og$\Alice{},\lft,$\og$ $\scriptstyle Alice ~\mapsto~ 1, Bob ~\mapsto~ 2$ $\fg$,\\
  \hspace*{65pt} $\og$\Bob{},\lft,$\og${\scriptsize $Alice \mapsto 3, Bob \mapsto 2$}$\fg$,$\og${\scriptsize$Alice \mapsto 2, Bob \mapsto
    2$}$\fg$,
  $\fg$.

Let us call $s_0$ the first strategy profile and $s_1$ the second one.   To describe an infinite strategy profile one uses most of the time a fixpoint equation like:
\[t \quad = \quad \og Alice, \lft, \og {\scriptstyle Alice ~\mapsto~ 0, Bob ~\mapsto~ 0}\fg, \og Bob, \lft, t, t\fg\fg\]
which corresponds to the pictures:
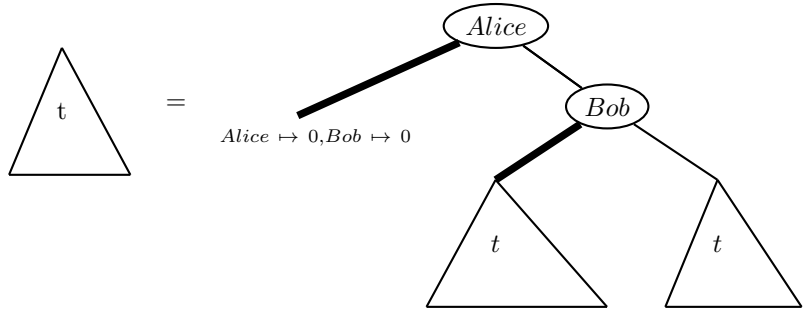
\[
\raisebox{50pt}{
  \begin{psmatrix}[colsep=6pt,rowsep=12pt]
  &&&&[name=0]\\
  &&&&t\\
  &[name=1]&&&&&&&[name=2]
  \ncline{0}{1}
  \ncline{0}{2}
  \ncline{1}{2}
\end{psmatrix}
}
\quad \raisebox{75pt}{=} \ \
\begin{psmatrix}[colsep=6pt,rowsep=12pt]
  &&& {\ovalnode{sommet}{Alice}}\\
  & [name=1]{\raisebox{-10pt}{$\scriptstyle Alice ~\mapsto~ 0, Bob ~\mapsto~ 0$}} &&& {\ovalnode{2}{Bob}}\\
  &&&[name=21]&&&&& [name=22]\\
  &&& t &&&&& t\\
  && [name=211] && [name=212] & [name=221] &&&&&&&& [name=222] & 
  \ncline[linewidth=.1]{sommet}{1}
  \ncline{sommet}{2}
  \ncline[linewidth=.1]{2}{21}
  \ncline{sommet}{2}
  \ncline{2}{22}
  \ncline{21}{211}
  \ncline{21}{212}
  \ncline{211}{212}
  \ncline{22}{221}
  \ncline{22}{222}
  \ncline{221}{222}
\end{psmatrix}
\]

Other examples of infinite strategy profiles are given in \Chapsec~\ref{chap:case_studies}.
Usually an infinite game is defined as a cofixpoint, \ie as the solution of an equation,
possibly a parametric equation.

Whereas in the finite case we can easily associate with a strategy profile a utility
function, \ie a function which assigns a utility to an agent, as the result of a recursive
evaluation, this is no more the case with infinite strategy profiles.  One reason is that
we are not sure that such a utility function exists for the strategy profile.  This makes
the function partial.  Therefore $s2u$ (an abbreviation for \emph{strategy
  profile-to-utility}) is a relation between a strategy profile and a utility function,
which is also a coinductive; $s2u$ appears in expression of the form $(s2u\ s\ a\ u)$
where $s$ is a strategy profile, $a$ is an agent and $u$ is a utility.  It reads ``$u$ is
a utility of the agent $a$ in the strategy profile~$s$''.

\begin{quotation}
$s2u$ is a predicate defined \textbf{coinductively} as follows:
  \begin{itemize}
  \item $s2u \og f \fg~a~(f(a))$ holds,
  \item if $s2u~s_l~a~u$ holds then $s2u~\og a',\lft,s_l,s_r\fg~a~u$ holds,
  \item if $s2u~s_r~a~u$ holds then $s2u~\og a',\rgt,s_l,s_r\fg~a~u$ holds.
  \end{itemize} 
\end{quotation}

This means the utility of $a$ for the leaf strategy profile $\og f\fg$ is $f(a)$, \ie the
value delivered by the function $f$ when applied to $a$.  The utility of $a$ for the
strategy profile $\og a',\lft,s_l,s_r\fg$ is $u$ if the utility of $a$ for the strategy
profile~$s_l$ is~$u$. For $s_0$, the first above strategy profile, one has
$s2u~s_0~\Alice{}~2$, which means that, for the strategy profile~$s_0$, the utility of
\Alice{} is $2$. 

\sssection{The predicate \ltl}{The predicate \ltl}

 In order to insure that $s2u$ has a result we define a predicate
\ltl{} that says that if one follows the choices shown by the strategy profile one reaches
a leaf, i.e., one does not go forever.

\begin{quotation}
The predicate \ltl{} is defined \textbf{inductively} as
  \begin{itemize}
  \item the strategy profile $\og f \fg$ \ltl{},
  \item if $s_l$ \ltl{}, then $\og a,\lft s_l, s_r\fg$ \ltl{},
  \item if $s_r$ \ltl{}, then $\og a, \rgt, s_l, s_r\fg$ \ltl{}.
  \end{itemize}
\end{quotation}

This means that a strategy profile which is itself a leaf \ltl{} and if the strategy
profile is a node, if the choice is \lft{} and if the left strategy subprofile \ltl{} then
the whole strategy \ltl{} and similarly if the choice is \rgt.  We claim that this gives a
good notion of \emph{finite horizon} which seems to be rather a concept on strategy
profiles than on games.

If $s$ is a strategy profile that satisfies the predicate \ltl{} then the utility exists
and is unique, in other words:
\begin{itemize}
\item[$`(!)$] \emph{Existence.}  For all agent $a$ and for all strategy profile $s$, if $s$ \ltl{} then there exists a utility $u$ which ``is a utility of the agent $a$ in the strategy profile~$s$''.
\item[$`(!)$] \emph{Uniqueness.}  For all agent $a$ and for all strategy profile $s$, if
  $s$ \ltl{}, if ``$u$ is a utility of the agent $a$ in the strategy~$s$'' and ``$v$ is a
  utility of the agent $a$ in the strategy~$s$'' then $u=v$.
\end{itemize}
We say ``the'' utility in this case since the relation $s2u~a$ is functional.   Now, with
an abuse of notation, we will write $s2u(s)(a)$ the utility of the agent $a$ in the
strategy profile $s$, when $s$ \textit{``leads to a leaf'}.

\sssection{The predicate \altl}{The predicate \altl}

We also consider a predicate \altl{} which means that everywhere in the strategy profile, if one follows the choices, one leads to a leaf.  This property is defined everywhere on an
infinite strategy profile and is therefore coinductive.

\begin{quotation}
  The predicate \altl{} is defined \textbf{coinductively}
  \begin{itemize}
  \item the strategy profile  $\og f \fg$ \altl{},
  \item for all choice $c$, if $\og a, c, s_l, s_r\fg$ \ltl{}, if $s_l$ \altl{}, if $s_r$ \altl{}, then $\og a, c, s_l, s_r\fg$ \altl{}.
  \end{itemize}
\end{quotation}

This says that a strategy profile, which is a leaf, \altl{} and that a composed strategy
profile inherits the predicate from its strategy subprofiles provided itself \ltl{}.

\sssection{The $\Box$ modality}{The $\Box$ modality}

$\Box$ is a modality, borrowed form temporal logic, \ie an operator which modifies a
predicate.  $\Box P$  reads \emph{always} $P$. 
\begin{quotation}
  The modality $\Box$ is defined \textbf{coinductively} by
  \begin{itemize}
  \item $P \og f \fg \  "=>" \  (`[ ] P) \og f \fg$
  \item $P~\og a, c,  s_l', s_r'\fg \  "=>" \  (`[ ] P) s_l \  "=>" \  (`[ ]
    P) s_f \  "=>" \  {(`[ ] P)~\og a, c,  s_l', s_r'\fg} $
  \end{itemize}
\end{quotation}
One has the proposition:
\begin{prop}
  $\forall s, (`[ ] \ltl)~s \ "<=>" \ s \altl$.
\end{prop}

\sssection{Bisimilarity}{The bisimilarity}

We define also bisimilarity between games and between strategy profiles.  For strategy
profiles, this is defined by:\pagebreak[3]

\begin{quotation}
  The bisimilarity $\sbis$ on \emph{strategy profiles} is defined \textbf{coinductively} as follows:
  \begin{itemize}
  \item $\og f\fg \ \sbis\ \og f\fg$,
  \item if $s_l\sbis s_l'$ and $s_r \sbis s_r'$ then $\og a, c, s_l s_r\fg\ \sbis\ \og a,c, s_l', s_r'\fg$.
  \end{itemize}
\end{quotation}

This says that two leaves are bisimilar if and only if they have the same utility function
and that two strategy profiles are bisimilar if and only if they have the same head agent,
the same choice and bisimilar strategy subprofiles.

\sssection{Game of a strategy profile}{The game of a strategy profile}

We can associate with a strategy profile a game  that is the game underlying the strategy
profile.  In other words, $s2g (s)$ is the game in which   all the choices are removed. 
\begin{quotation}
  The fonction $s2g$ is defined \textbf{coinductively} as follows:
  \begin{itemize}
  \item $s2g \og f \fg \quad= \quad \ogame f \fgame$
  \item $s2g \og a, c, s_l, s_r\fg \quad = \quad \ogame a, s2g (s_l), s2g (s_r) \fgame$.
  \end{itemize}
\end{quotation}

\ssection{Equilibria}{Subgame perfect equilibria and Nash equilibria}
\label{sec:equi}

Nash equilibria are specific strategy profiles, but to define them one needs the concept
of convertibility.

\sssection{Convertibility}{Convertibility}

Despite it is not strictly defined in textbooks as such, \emph{convertibility} is an
important binary relation on strategy profiles, necessary to speak formally about
equilibria.  Indeed in order to characterize a strategy profile $e$ as a Nash equilibria,
it is assumed that each agent compares the payoff returned by that strategy profile $e$
with the payoff returned by other strategy profiles, which are ``converted'' from the
$e$  by the agent changing his mind.  Since this relation plays a crucial role in
formal definition of a Nash equilibrium, it is worth describing, first informally, then a
little more formally, knowing that the ultimate formal definition is given in the in
\Coq{} vernacular on page~\pageref{sec:excerpts-coq-devel}.  Convertibility was
introduced for finite games by \citet{vestergaard06:IPL}.   

\ssssection{Convertibility informally}

\citet[Chap. 5]{osborne04a} presents a Nash equilibrium as ``a strategy profile from which
no player wishes to deviate, given the other player's strategies''.  We have therefore to
say what one means by ``a player deviating when the others do not''.  In other words, we
want to make precise the concept of ``deviation''.  For that, assume given an agent~$a$
and a strategy profile~$s$, a strategy profile in which only finitely many choices made by
the given agent $a$ are changed is a ``deviation'' of~$a$ and is said, in this framework,
to be convertible to $s$ for~$a$.  The binary relation between two strategy profiles,
which we call \emph{convertibility}, can be made precise by giving it an inductive
definition.\footnote{It is also possible to give it a coinductive definition, in which
  infinitely many choices can be changed, but we feel that this goes beyond the ability of
  a rational agent who has finite capacities to reason.}  In the previous examples, $s_0$
is convertible to $s_1$ for $\Alice$, since the only change between $s_0$ and $s_1$ is
$\Alice$ changing her first choice.

\ssssection{Convertibility as an inductively defined  mathematical relation}

We write $\conva$  the convertibility for agent $a$.

\begin{quotation}
  The relation $\conva$ is defined \textbf{inductively} as follows:
\ifElsev \begin{description}[$\star$] \else \begin{description} \fi
\item[ConvBis:] $\conva$ contains bisimilarity $\sbis$, i.e.,  
    \[ \prooftree s\sbis s'
    \justifies s \conva s'
    \endprooftree
    \]
  
  \item[ConvAgent]: If the node has the same agent as the agent in $\conva$ then the
    choice may change, i.e.,
    \[
    \prooftree s_1 \conva s_1'\qquad s_2 \conva s_2' %
    \justifies \og a, c, s_1, s_2 \fg \ \conva \ \og a, c', s_1', s_2'\fg
    \using ConvAgent
    \endprooftree
    \]

  \item[ConvChoice:] If the node does not have the same agent as in $\conva$, then the choice has to be the same:
    \[
    \prooftree 
  s_1 \conva s_1'\qquad s_2 \conva s_2' %
  \justifies \og a', c, s_1, s_2 \fg \ \conva \ \og a', c, s_1', s_2'\fg
    \using ConvChoice
  \endprooftree
  \]
\end{description}
\end{quotation}

Roughly speaking two strategy profiles are convertible for $a$ if their difference only
for the choices of $a$.  In the previous example (Section~\ref{sec:inf_stra}) we may write
$s_0 \convalice s_1$, to say that $s_0$ is convertible to $s_1$ for $\Alice$.  In
Figure~\ref{fig:ind_proof_conv}, page~\pageref{fig:ind_proof_conv}, we develop the
skeleton of the proof of this convertibility, namely we give the proof tree of
$s_0~\convalice~s_1$.  Since $\conva$ is defined inductively, this means that the changes
are finitely many.  We feel that this makes sense since an agent can only conceive
finitely many issues.  For instance for two strategy profiles associated with the $0,1$ game, we
get

\pagebreak[2]
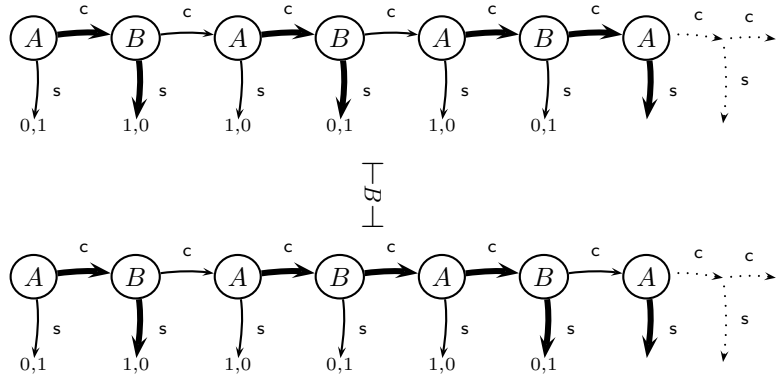
\bigskip

\begin{math}
  \begin{psmatrix}[colsep=20pt,rowsep=20pt] 
    & {\ovalnode{a}{A}}
    &{\ovalnode{b}{B}} %
    & {\ovalnode{c}{A}} & {\ovalnode{d}{B}} &{\ovalnode{e}{A}}
    &{\ovalnode{f}{B}} & {\ovalnode{g}{A}} & [name=h] & [name=i] &\\ 
    &[name=a1] {\scriptstyle 0,1} %
    &[name=b1] {\scriptstyle 1,0} %
    & [name=c1]{\scriptstyle 1,0} %
    &[name=d1] {\scriptstyle 0,1} %
    &[name=e1] {\scriptstyle 1,0} %
    &[name=f1] {\scriptstyle 0,1} %
    &[name=g1] \phantom{\scriptstyle 1, 0} %
    & [name=h1] %
    \ncarc[arrows=->,linewidth=.08]{a}{b}\Aput{\scriptstyle \mathsf{c}} %
    \ncarc[arrows=->]{b}{c} \Aput{\scriptstyle \mathsf{c}} %
    \ncarc[arrows=->,linewidth=.08]{c}{d} \Aput{\scriptstyle \mathsf{c}} %
    \ncarc[arrows=->]{d}{e} \Aput{\scriptstyle \mathsf{c}} %
    \ncarc[arrows=->,linewidth=.08]{e}{f} \Aput{\scriptstyle \mathsf{c}} %
    \ncarc[arrows=->,linewidth=.08]{f}{g} \Aput{\scriptstyle \mathsf{c}} %
    \ncarc[arrows=->,linestyle=dotted]{g}{h}\Aput{\scriptstyle \mathsf{c}} %
    \ncarc[arrows=->,linestyle=dotted]{h}{i} \Aput{\scriptstyle \mathsf{c}} %
    \ncarc[arrows=->]{a}{a1} \Aput{\scriptstyle \mathsf{s}}%
    \ncarc[arrows=->,linewidth=.08]{b}{b1} \Aput{\scriptstyle \mathsf{s}} %
    \ncarc[arrows=->]{c}{c1} \Aput{\scriptstyle \mathsf{s}} %
    \ncarc[arrows=->,linewidth=.08]{d}{d1} \Aput{\scriptstyle \mathsf{s}} %
    \ncarc[arrows=->]{e}{e1} \Aput{\scriptstyle \mathsf{s}} %
    \ncarc[arrows=->]{f}{f1} \Aput{\scriptstyle \mathsf{s}} %
    \ncarc[arrows=->,linewidth=.08]{g}{g1}\Aput{\scriptstyle \mathsf{s}} %
    \ncarc[arrows=->,linestyle=dotted]{h}{h1} \Aput{\scriptstyle \mathsf{s}}%
  \end{psmatrix}
\end{math}

\bigskip

\centerline{\large $\mathbf{\top}$}
\centerline{\it \small \raisebox{12pt}{\rotatebox{270}{B}}}
\vspace*{-8pt}
\centerline{\large $\bot$}

\bigskip

\begin{math}
  \begin{psmatrix}[colsep=20pt,rowsep=20pt] 
    & {\ovalnode{a}{A}}
    &{\ovalnode{b}{B}} %
    & {\ovalnode{c}{A}} & {\ovalnode{d}{B}} &{\ovalnode{e}{A}}
    &{\ovalnode{f}{B}} & {\ovalnode{g}{A}} & [name=h] & [name=i] &\\ 
    &[name=a1] {\scriptstyle 0,1} %
    &[name=b1] {\scriptstyle 1,0} %
    & [name=c1]{\scriptstyle 1,0} %
    &[name=d1] {\scriptstyle 0,1} %
    &[name=e1] {\scriptstyle 1,0} %
    &[name=f1] {\scriptstyle 0,1} %
    &[name=g1] \phantom{\scriptstyle 1, 0} %
    & [name=h1] %
    \ncarc[arrows=->,linewidth=.08]{a}{b}\Aput{\scriptstyle \mathsf{c}} %
    \ncarc[arrows=->]{b}{c} \Aput{\scriptstyle \mathsf{c}} %
    \ncarc[arrows=->,linewidth=.08]{c}{d} \Aput{\scriptstyle \mathsf{c}} %
    \ncarc[arrows=->,linewidth=.08]{d}{e} \Aput{\scriptstyle \mathsf{c}} %
    \ncarc[arrows=->,linewidth=.08]{e}{f} \Aput{\scriptstyle \mathsf{c}} %
    \ncarc[arrows=->]{f}{g} \Aput{\scriptstyle \mathsf{c}} %
    \ncarc[arrows=->,linestyle=dotted]{g}{h}\Aput{\scriptstyle \mathsf{c}} %
    \ncarc[arrows=->,linestyle=dotted]{h}{i} \Aput{\scriptstyle \mathsf{c}} %
    \ncarc[arrows=->]{a}{a1} \Aput{\scriptstyle \mathsf{s}}%
    \ncarc[arrows=->,linewidth=.08]{b}{b1} \Aput{\scriptstyle \mathsf{s}} %
    \ncarc[arrows=->]{c}{c1} \Aput{\scriptstyle \mathsf{s}} %
    \ncarc[arrows=->]{d}{d1} \Aput{\scriptstyle \mathsf{s}} %
    \ncarc[arrows=->]{e}{e1} \Aput{\scriptstyle \mathsf{s}} %
    \ncarc[arrows=->,linewidth=.08]{f}{f1} \Aput{\scriptstyle \mathsf{s}} %
    \ncarc[arrows=->,linewidth=.08]{g}{g1}\Aput{\scriptstyle \mathsf{s}} %
    \ncarc[arrows=->,linestyle=dotted]{h}{h1} \Aput{\scriptstyle \mathsf{s}}%
  \end{psmatrix}
\end{math}

\sssection{Nash equilibria}{Nash equilibria}

The notion of Nash equilibrium is translated from the notion in textbooks.  Let us recall
it. According to \citet[chap. 5]{osborne04a}, \textit{A Nash equilibrium is a``pattern[s]
  of behavior with the property that if every player knows every other player's behavior
  she has not reason to change her own behavior''} in other words, \textit{``a Nash
  equilibrium [is] a strategy profile from which no player wishes to deviate, given the
  other player's strategies.'' }  As we said, the informal concept of deviation is
expressed formally by the binary relation ``convertibility'.  The concept of Nash
equilibrium is based on a comparison of utilities.
$s$ is a \emph{Nash equilibrium} if the following implication holds:
\begin{quotation}
  \noindent If for all agent~$a$ and for all strategy profile~$s'$ which is convertible to
  $s$, i.e., $s\conva s'$, if $u$ is the utility of~$s$ for~$a$ and $u'$ is the utility
  of~$s$' for $a$, then $u' \leut u$.
\end{quotation}
Roughly speaking this means that a Nash equilibrium is a strategy profile in which no
agent has interest to change his choice since doing so he cannot get a better payoff.

\sssection{Subgame Perfect Equilibria}{Subgame Perfect Equilibria}
\label{sec:sgpe}

Let us consider now \emph{subgame perfect equilibria}, which we write $SGPE$.  $SGPE$ is a
property of strategy profiles. It requires the strategy subprofiles to fulfill
coinductively the same property, namely to be a $SGPE$, and to insure that the strategy
profile with the best utility for the node agent to be chosen.  Since both the strategy
profile and its strategy subprofiles are potentially infinite, it makes sense to define
$SGPE$ coinductively.

\pagebreak[2]
\begin{quotation}
  $SGPE$ is defined \textbf{coinductively} as follows:
\ifElsev \begin{itemize}[$\star$] \else \begin{itemize} \fi
  \item $SGPE \og f\fg$,
  \item if $\og a,\lft,s_\lft,s_r\fg$ \altl{}, if $SGPE(s_\lft)$ and $SGPE(s_r)$, if
    $s2u(s_r)(a) \leut s2u(s_\lft)(a)$ then $SGPE~\og a,\lft,s_\lft,s_r\fg$,
  \item if $\og a,\rgt{},s_\lft,s_r\fg$ \altl{}, if $SGPE(s_\lft)$ and $SGPE(s_r)$, if
    $s2u(s_\lft)(a) \leut s2u(s_r)(a)$then $SGPE~\og a,\rgt{},s_\lft,s_r\fg$.
  \end{itemize}
\end{quotation}

This means that a strategy profile, which is a leaf, is a subgame perfect equilibrium.
Moreover if the strategy profile is a node, if the strategy profile \altl{}, if it has
agent $a$ and choice \lft, if both strategy subprofiles are subgame perfect equilibria and
if the utility of the agent $a$ for the right strategy subprofile is less than this for
the left strategy subprofile then the whole strategy profile is a subgame perfect
equilibrium and vice versa.  If the choice is \rgt{} this works similarly.

Notice that since we require that the utility can be computed not only for the strategy
profile, but for the strategy subprofiles and for the strategy subsubprofiles and so on,
we require these strategy profiles not only to \textit{``lead to a leaf''} but to
\textit{``always lead to a leaf''}.

We define orders (one for each agent $a$) between strategy profiles that \emph{lead to a
  leaf} which we write $\le_a$.
\begin{quotation}
 \noindent $s' \le_{a} s$ iff  $s2u(s')(a) \leut s2u(s)(a)$.
\end{quotation}
We say ``the'' utility since in this case the relation $s2u~a$ is functional. 

\begin{prop}
$\le_a$ is an order on strategy profiles which \textit{lead to a leaf}.
\end{prop}

The proof is straightforward.

\begin{prop}
  A subgame perfect equilibrium is a Nash equilibrium.
\end{prop}

\begin{proof}
  Suppose that $s$ is a strategy profile which is a $SGPE$ and which has to be proved to
  be is a Nash equilibrium.

Assuming that $s'$ is a strategy profile such that $s \conva s'$, let us prove by
induction on $s \conva s'$ that $s' \le_{a} s$:
\begin{itemize}
\item Case $s = s'$, by reflexivity, $s' \le_{a} s$.
\item Case $s = \og x, \lft ,s_\lft,s_r\fg$ and $s' = \og x, \lft , s'_\lft, s'_r \fg$ with
${x\neq a}$. $s \conva s'$ and the definition of $\conva$ imply $s_\lft \conva s'_\lft$ and $s_r
\conva s'_r$.  $s_\lft$ which is a strategy subprofile of a $SGPE$ is a $SGPE$ as well.
Hence by induction hypothesis, $s'_\lft \le_a s_\lft$.

The utility of $s$ (respectively of $s'$) for $a$ is the utility of $s_\lft$ (respectively of
$s_\lft'$) for $a$, then $s' \le_a s$.
\item The case $s \mathop{=} \og x, \rgt{} ,s_\lft,s_r\fg$ and $s' \mathop{=} \og x, \rgt{} ,
s'_\lft, s'_r \fg$ is similar.
\item Case $s \mathop{=} \og a, \lft, s_\lft, s_r \fg$ and $s' \mathop{=} \og a, \rgt{},
s'_\lft, s'_r \fg$, then $s_\lft \conva s'_\lft$ and $s_r \conva s'_r$.  Since $s$ is a $SGPE$,
$s_r \le_a s_\lft$.

Moreover, since $s_r$ is a $SGPE$, by induction hypothesis, $s'_r~\le~s_r$.  Hence, by
transitivity of $\le_a$, $s'_r \le_a s_\lft$.  But we know that the utility of $s'$ for $a$
is the one of $s'_r$ and the utility of $s$ for $a$ is the one of~$s_\lft$, hence $s' \le_a
s$.
\item The case $s \mathop{=} \og a, \rgt{}, s_\lft, s_r \fg$ and $s' \mathop{=} \og a, \lft,
s'_\lft, s'_r \fg$ is similar.

\end{itemize}
\end{proof}
The above proof is a presentation of the formal proof written with the help of the proof
assistant \Coq.  Notice that it is by induction on $\conva$ which is possible since
$\conva$ is inductively defined.  Notice also that $s$ and $s'$ are potentially infinite.

\ifBook
\begin{landscape}
\begin{figure}
  \centering
     \vspace*{60pt}
  \begin{tiny}
    \prooftree \prooftree \prooftree \justifies\og{ \Alice \mapsto 0, \Bob \mapsto 1}\fg
    \convalice \og{\Alice \mapsto 0, \Bob \mapsto 1}\fg \using \textit{ConvRefl}
    \endprooftree
    \ \ \prooftree \justifies\og{ \ldots}\fg \convalice \og{ \ldots}\fg \using
    \textit{ConvRefl}
    \endprooftree
    \justifies \og \emph{\Bob}, \lft, \og{ \ldots}\fg \og{ \ldots}\fg \fg%
    \convalice %
    \og\emph{\Bob}, \lft, \og\ldots\fg \og{ \ldots}\fg \fg \using \textit{ConvAgent}
    \endprooftree
    \ \ \ \prooftree \justifies \og{\Alice \mapsto 2, \Bob \mapsto 0}\fg \convalice
    \og{\Alice \mapsto 2, \Bob \mapsto 0}\fg \using \textit{ConvRefl}
    \endprooftree
    \justifies \og\emph{\Alice},\lft,\og\emph{\Bob}, \lft, \og{ \ldots}\fg, %
    \og{\ldots}\fg \fg, \og \ldots \fg\fg \convalice \og\emph{\Alice},\rgt,\og\emph{\Bob},
    \lft, \og{ \ldots}\fg, \og{\ldots}\fg \fg, \og \ldots \fg\fg \using
    \textit{ConvChoice}
    \endprooftree
  \end{tiny}
  \caption{An inductive proof of convertibility}
  \label{fig:ind_proof_conv}
\end{figure}
\end{landscape}
\fi
\ssection{Escalation}{Escalation}
\label{sec:escalation}

A game is susceptible to escalation or not. Obviously the possibility of an escalation in
a game requires the game to be infinite.

\sssection{Escalation informally}{Escalation informally}

Escalation is a property of an infinite game, which says that a game can contain an
infinite path along which players always act rationally.  In other words, it says that
at each turn in the game, there exists a strategy profile which is a subgame perfect
equilibrium,  in which the player who has the turn continues.  Since at each turn,
continuing is rational for each player, this means that there is a possibility for players
acting rationally to continue forever.  That there is an infinite path means that there exists
an infinite sequence of games which are direct subgames of their predecessors.  For each
game of this sequence, there exists a strategy profile which has this game as a skeleton,
which is a subgame perfect equilibrium and in which the player who has the turn continues.

\sssection{Escalation as a formal property}{Escalation as a formal property}

The property of having an escalation can be formalized.  A game $g$ \emph{has an
  escalation} if there exists an escalation sequence~$(\gs_n)_{n\in \nat}$ which is a
sequence of subgames of $g$ along the escalation with the following properties: for all
$n$ there are two strategy profiles $s$ and $s'$ and an agent $a$ such that
\begin{itemize}
\item the game associated with the strategy profile $\og a,\lft, s, s'\fg$ is bisimilar to the game
  $\gs_{n}$ and is a \emph{subgame perfect equilibrium} or the game associated with the strategy profile $\og a,\rgt, s', s\fg$ is bisimilar to the game
  $\gs_{n}$ and is a \emph{subgame perfect equilibrium} 
\item the game associated with the strategy profile $s$ is bisimilar to the game
  $\gs_{n+1}$, (this insures that $\gs_{n+1}$ is bisimilar to a direct subgame of $\gs_n$, more
  precisely that $\gs_n$ is bisimilar to the game $\langle a, \gs_{n+1}, s2g(s')\rangle$ or to
  the game $\langle a, s2g(s'), \gs_{n+1}\rangle$, according to the choice made in the above condition).
\end{itemize}
This can be made completely formal, by writing it in \Coq:

\noindent
\coqdockw{Definition} \coqdocid{has\_an\_escalation\_sequence}
(\coqdocid{g\_seq}:\coqdocid{nat} \ensuremath{\rightarrow} \coqdocid{Game}):
\coqdocid{Prop} := \\
\ensuremath{\forall} \coqdocid{n}:\coqdocid{nat}, 
\ensuremath{\exists} \coqdocid{s},  \ensuremath{\exists} \coqdocid{s'}, \ensuremath{\exists} \coqdocid{a}, \coqdoceol
\coqdocindent{1.00em}
(\coqdocid{s2g} \og\coqdocid{a},\coqdocid{l},\coqdocid{s},\coqdocid{s'}\fg
~$\gbis$~ \coqdocid{g\_seq} \coqdocid{n} \ensuremath{\land} \coqdocid{SGPE}
\og\coqdocid{a},\coqdocid{l},\coqdocid{s},\coqdocid{s'}\fg \ \ensuremath{\lor}\coqdoceol
\coqdocindent{1.00em}
\coqdocid{s2g} \og\coqdocid{a},\coqdocid{r},\coqdocid{s'},\coqdocid{s}\fg
~$\gbis$~ \coqdocid{g\_seq} \coqdocid{n} \ensuremath{\land} \coqdocid{SGPE}
\og\coqdocid{a},\coqdocid{r},\coqdocid{s'},\coqdocid{s}\fg) \ \ensuremath{\land}\coqdoceol
\coqdocindent{1.00em}
\coqdocid{s2g} \coqdocid{s} ~$\gbis$~ \coqdocid{g\_seq} (\coqdocid{n}+1).\coqdoceol

\medskip
\noindent
\coqdockw{Definition} \coqdocid{has\_an\_escalation} (\coqdocid{g}:\coqdocid{Game}) : \coqdocid{Prop} :=\coqdoceol
\coqdocindent{2.00em}
\ensuremath{\exists} \coqdocid{g\_seq}, (\coqdocid{has\_an\_escalation\_sequence} \coqdocid{g\_seq}) \ensuremath{\land} (\coqdocid{g\_seq} 0 = \coqdocid{g}).\coqdoceol
\coqdocindent{1.00em}
\coqdoceol

\chapitre{Case studies}{Case studies}
\label{chap:case_studies}

In this section we study several kinds of games that have some analogies, especially they
have a centipede shape, since they have an infinite backbone (on the ``left'') and all the
right subgames are leaves.  In the two last cases, the utilities go to infinity, but in
the second (dollar auction game) the utilities go to $(-\infty,-\infty)$ (costs, \ie the
opposites of utilities, go to $(+\infty,+\infty)$), whereas in the third, (infinipede
game) the utilities go to $(+\infty,+\infty)$.   

\ssection{The $0, 1$ game}{The $0, 1$ game}
\label{sec:0-1-game}\index{$0,1$ game}

In the $0,1$ game (Figure~\ref{fig:zone}, p.~\pageref{fig:zone}) $0$ and $1$ are
payoffs. The $0,1$ game has many subgame perfect equilibria, namely the strategy profiles
in which Alice continues always and Bob stops infinitely often and the strategy profiles
in which Bob continues always and Alice stops infinitely often.

\sssection{Two simple subgame perfect equilibria}{Two simple subgame perfect equilibria}
\label{sec:two-simple-subgame}

For what we are interested in, we can consider two strategy profiles, one in each
category:
\begin{itemize}
\item the strategy profile \emph{``\Alice{} continues always and \Bob{} stops always''}, which
  we call \textsf{z1AcBs} and,
\item the strategy profile \emph{``\Alice{} stops always and \Bob{} continues always''}, which
  we call \textsf{z1AsBc}.
\end{itemize}
The reasoning to show that \textsf{z1AcBs} is a subgame perfect equilibrium works as
follows. In this strategy profile \Alice{} gets $1$ and \Bob{} gets $0$. Assume the strategy
subprofile of \textsf{z1AcBs} after the second turn for \Alice{}\footnote{which is nothing
  but \textsf{z1AcBs}!} is a subgame perfect equilibrium, for which \Alice{} gets $1$ and \Bob{}
gets $0$.  The strategy subprofile that starts at \Bob{}'s turn and which we call
\textsf{sg\_z1AcBs} is a subgame perfect equilibrium for which \Alice{} gets $1$ and \Bob{} gets
$0$, since its two strategy subprofiles are subgame perfect equilibria for which \Alice{}
gets $1$ and \Bob{} gets~$0$.  \textsf{z1AcBs} is a subgame perfect equilibrium since its two
strategy subprofiles are two subgame perfect equilibria, one is \textsf{sg\_z1AcBs} for
which \Alice{} gets $1$ and the other is a leaf for which \Alice{} gets $0$.

The same reasoning applies to \textsf{z1AsBc} to prove that it is a subgame perfect
equilibrium.

\sssection{Cutting the game and extrapolating}{Cutting the game and extrapolating}
\label{sec:cutting-game}

If one cuts the $0,1$~game at a finite position, to obtain a finite game, one can cut either
after \Alice{} like on the left below or on can cut after \Bob{} like on the right below:

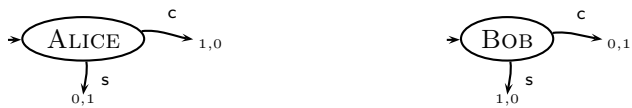
\[\begin{psmatrix}[colsep=5pt,rowsep=9pt]
& [name=z]& {\ovalnode{a}{\scriptscriptstyle \Alice}} && [name=b] {\quad\scriptscriptstyle 1,0} \\
&& [name=c]  {\scriptscriptstyle 0,1} 
\ncline[arrows=->]{z}{a} 
\ncarc[arrows=->]{a}{b} \Aput{\scriptstyle \mathsf{c}}
\ncarc[arrows=->]{a}{c} \Aput{\scriptstyle \mathsf{s}}
\end{psmatrix}
\qquad \qquad\qquad \qquad
\begin{psmatrix}[colsep=5pt,rowsep=9pt]
& [name=z]& {\ovalnode{a}{\scriptscriptstyle \Bob}} && [name=b] {\quad\scriptscriptstyle 0,1} \\
&& [name=c]  {\scriptscriptstyle 1,0} 
\ncline[arrows=->]{z}{a} 
\ncarc[arrows=->]{a}{b} \Aput{\scriptstyle \mathsf{c}}
\ncarc[arrows=->]{a}{c} \Aput{\scriptstyle \mathsf{s}}
\end{psmatrix}
\]
When one cuts after \Alice, the backward induction equilibrium is when \Alice{} continues
always and \Bob{} does whatever he wants and when one cuts after \Bob, the backward
induction equilibrium is when \Alice{} does whatever she wants and \Bob{} continues
always.  One sees that those equilibria cannot be extrapolated to the infinity since they
are inconsistent except in the only case when \Alice{} and \Bob{} continue forever.  But this
strategy profile cannot be a subgame perfect equilibrium, since it does not fulfill the
predicate \altl.  Hence cutting the game at
a finite position gives no clue on what one obtains at the limit on the infinite game.

\sssection{The $0,1$ game has a rational escalation}{The $0,1$ game has a rational escalation}
\label{sec:zone-has-an}

One sees that \zone{} has a rational escalation. Indeed at each step, the agents can
always make the choice to continue which is rational since this corresponds to the first
choice of a subgame perfect equilibrium.  If both agents make the choice to continue, this
is the escalation.  \Alice{} continues since she feels that despite \Bob{} did not stop,
he will eventually stop and \Bob{} continues because he feels that \Alice{} will
eventually stop.

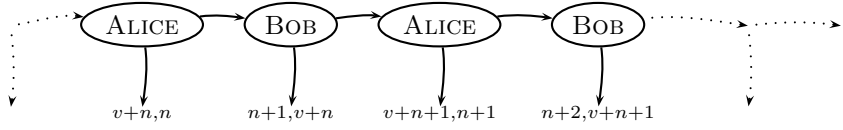
\ssection{The dollar auction}{The dollar auction game}
\label{sec:dol}\index{dollar auction}

\index{escalation}
The dollar auction has been presented by \citet{Shubik:1971} \index{Shubik} as the
paradigm of escalation, insisting on its paradoxical aspect. It is a sequential game
presented as an auction in which two agents compete to acquire an object of value $v$
($v>0$) (see \citet[Ex.  3.13]{gintis00:_game_theor_evolv}).  Suppose that both agents bid
$\$ 1$ at each turn. If one of them gives up, the other receives the object and both pay
the amount of their bid.\footnote{In a variant, each bidder, when he bids, puts a dollar
  bill in a hat or in a piggy bank and their is no return at the end of the auction.  The
  last bidder gets the object.}  For instance, if agent \Alice{} stops immediately, she
pays nothing and agent \Bob{}, who acquires the object, has a payoff $v$.  In the general
turn of the auction, if \Alice{} abandons, she looses the auction and has a payoff $-n$
and \Bob{} who has already bid $-n$ has a payoff $v-n$.  At the next turn after \Alice{}
decides to continue, bids $\$ 1$ for this and acquires the object due to \Bob{} stopping,
\Alice{} has a payoff $v-(n+1)$ and \Bob{} has a payoff $-n$.  In our formalization we
have considered the \emph{dollar auction} up to infinity.  Since we are interested only by
the ``asymptotic'' behavior, we can consider the auction after the value of the object has
been passed and the payoffs are negative.  The dollar auction game can be summarized by
Figure~\ref{fig:dol_auct}.  Notice that we assume that {\Alice} starts.  

\sssection{Equilibria in the dollar auction}{Equilibria in the dollar auction}
\label{sec:equil-doll-auct}

We have
recognized three classes of infinite strategy profiles, indexed by~$n$:

\figdollar

\begin{enumerate}
\item The strategy profile \emph{always give up}, in which both {\Alice} and \Bob{} stop at each turn, in short \textsf{dolAsBs}$_n$.
\item The strategy profile \emph{{\Alice} stops always and \Bob{} continues always}, in short \textsf{dolAsBc}$_n$.
\item The strategy profile \emph{\Alice{} continues always and \Bob{} stops always}, in short \textsf{dolAcBs}$_n$.
\end{enumerate}
The three kinds of strategy profiles are presented in Figure~\ref{fig:4_strat}.

\figfourstrat

In the figures like in the \Coq{} implementation, we use \emph{costs}\footnote{Recall that
  cost is the opposite of utility.} instead of payoffs or utilities, since it is simpler in
the \Coq{} formalization to reason on natural numbers.

We have shown\footnote{The proofs are typical uses of the \Coq{} \textsf{cofix} tactic. } that the second and third kinds of strategy profiles, in
which one of the agents always stops and the other continues, are subgame perfect equilibria.  For instance, consider the strategy profile
\textsf{dolAsBc}$_n$.  Assume $\mathit{SGPE}(\textsf{dolAsBc}_{n+1})$.  It works as follows: if $\textsf{dolAsBc}_{n+1}$ is a subgame perfect
equilibrium corresponding to the payoff ${-(v+n+1), -(n+1)}$, then
\[\og \Bob, \lft, \textsf{dolAsBc}_{n+1}, \og \Alice "|->" -(n+1), \Bob "|->" -(v+n)\fg \fg\] 
is again a subgame perfect equilibrium (since  $-(n+1) \ge -(v+n)$) and therefore $\textsf{dolAsBc}_{n}$ which is

\medskip
\(\og \Alice, \rgt{}, \og\Bob, \lft, \textsf{dolAsBc}_{n+1}, \og \Alice "|->" -(n+1), \Bob "|->" -(v+n)\fg , \\
\hspace*{65pt} \og \Alice "|->" -(v+n), \Bob "|->" -n\fg \fg\)
\medskip

\noindent is a subgame perfect equilibrium, since again
$-(v+n) \ge -(v+n+1)$.\footnote{Since the
  \textsf{cofix} tactic has been used on a strict strategy subprofile, the reasoning is correct.}  We can conclude that for all $n$, $\textsf{dolAsBc}_{n}$ \emph{is a subgame
  perfect equilibrium}.  In other words, we have assumed that $\mathit{SGPE}(\textsf{dolAsBc}_{n})$ is an \emph{invariant} all along the game and that this invariant is preserved
as we proceed backward, through time, into the game.

With the condition $v>1$, we can prove that \textsf{dolAsBs}$_0$ is not a Nash
equilibrium, then as a consequence not a subgame perfect equilibrium.  Therefore, the
strategy profile that consists in stopping from the beginning and forever is not a Nash
equilibrium, this contradicts what is said in the literature
\citep*{Shubik:1971,oneill86:_inten_escal_and_dollar_auction,leininger89:_escal_and_coop_in_confl_situat,gintis00:_game_theor_evolv},
with a finite game to approximate an infinite game (the escalation).

\sssection{Escalation is rational}{Escalation is rational in the dollar auction}
\label{sec:esc}

Agents are rational when they choose at each step, what they feel to be their best option.
Many authors agree\footnote{See however
  \citep{halpern01:_subst_ration_backw_induc,stalnaker98:_belief_revis_in_games}.} that
rationality is choosing a subgame perfect equilibrium.  
\cite{aumann95} is one of the strongest advocate of this position.  His principle in the
dollar auction says that a rational agent will choose at each step one of the strategy
profiles which is a subgame perfect equilibrium, namely \textsf{dolAsBc}$_n$ or
\textsf{dolAcBs}$_n$.  Suppose \Alice{} is in the middle of the auction, she has two
options that are rational: one option is to take \Bob's threat seriously and to stop right
away, since she assumes that \Bob{} will continue always (strategy profile
$\textsf{dolAsBc}_n$).  But in her second option, she admits that from now on {\Bob} will
stop always (strategy profile $\textsf{dolAcBs}_n$) and she will always continue: this is
a subgame perfect equilibrium, hence she is rational.  If {\Bob} acts similarly this is
the escalation.  So at each step an agent can stop and be rational, as well as at each
step an agent can continue and be rational; both options make sense, and the escalation as well.

\AliceBob

We claim that human agents reason coinductively unknowingly.  Therefore, for them,
continuing always is one of their rational options at least if one considers strictly the
rules of the dollar auction game with no limit on the bankroll.  If at all steps, both
agents continue always, this is the escalation.  Many experiences (see
\citep{colman99:_game_theor_and_its_applic} for a survey) have shown that human are
inclined to escalate or at least to go very far in the auction when playing the dollar
auction game.  We propose the following explanation: the finiteness of the game was not
explicit for the participants and for them the game was naturally infinite.  Therefore
they adopted a kind of reasoning similar to the one we develop here, probably in an
intuitive form and they conclude it was equally rational to continue or to leave according
to their comprehension of the threat of their opponent.  Actually our theoretical work
reconciles experiences with logic,\footnote{A logic which includes coinduction.} and human
reasoning with rationality.

If agents would have a global view, they would notice that escalation is scarring and
ruining and should be avoided.  Since it is rational it looks ineluctable, but there are a
few ways to avoid it.  First, the game can be stopped by the action of an external
observer, a kind of master of ceremony who declares the game over, something similar to
the finite payroll.  Second, since it is a consequence of the absence of communication
between agents, it can be avoided, if the agents communicate, like Kennedy and Khrushchev
did with the Moscow-Washington hotline.  Third, escalation can be stopped by introducing
another challenge, like building Europe by Adenhauer and de Gaulle.

\ssection{The infinipede}{The infinipede}
\label{sec:anot-exampl-infin}

Often studied, the extensive game called \emph{centipede}\footnote{ A centipede has hundred legs, whereas a millipede has thousand. All belong to the
  group of myriapods which means ``ten thousand legs''.} has been introduced by \cite{rosenthal81:_games_of_perfec_infor_predat} (see also
\cite{binmore87:_model_ration_player,Colman_rat_back_ind,osborne94:_cours_game_theory}). In \citet[p. 96]{rosenthal81:_games_of_perfec_infor_predat}
the game is  pictured as shown in Figure~\ref{fig:rosen}.

\rosenthal

This finite game has one Nash equilibrium obtained by backward induction, namely by agent
$A$ stopping immediately.  This game has been extended to hundred, thousand nodes, but all
those extensions are finite and all authors conclude that there is a unique Nash
equilibrium in which the agents give up immediately.

Since we had noticed, in the case of $0,1$~game and the dollar auction, a discrepancy between Nash
equilibria when going from finite games to infinite games, it is challenging to see
whether the same phenomenon occurs when going from the \emph{centipede} (a finite game)
to the \emph{infinipede} (an infinite game).

The \emph{infinipede} is an infinite game in extensive form in which agent $A$ has the
choice to continue or to end the game with both agents receiving a payoff $2n$ and agent
$B$  has the choice to continue or to end the game with agent $A$ receiving
$2n-1$ and agent $B$ receiving $2n+3$.  The generalization is an infinite game which can
be pictured as follows (we use subtraction over naturals in which $2*0 - 1 = 0$):

\infinipede

In infinipedes, we have identified only one subgame perfect equilibrium, namely this where
both agents abandon at each turn (Figure~\ref{fig:infSGPE}).  We call it
\textsf{cent\_agu}~$n$.
\infSGPE
This shows that even in the infinite generalization, agents are rational if they do not
start the game and abandon from the beginning.  We have actually shown this property for
each value of $n$, namely for all $n$, \textsf{cent\_agu}~$n$ is a subgame perfect
equilibrium.  In \Coq{} we have proved the theorem:

\medskip
\noindent
\qquad \coqdockw{Theorem} \coqdocid{SGPE\_cent\_AGU}:  \ensuremath{\forall} (\coqdocid{n}:\coqdocid{nat}), \coqdocid{SGPE} \coqdocid{le} (\coqdocid{cent\_agu} \coqdocid{n}).\coqdoceol
\medskip
\noindent Hence the paradox remains: the agents do not get the somewhat better payoff,
they would get if they would be more flexible with respect to rationality.  However, there
remains a problem for the agents in the infinipede game: when they start an infinite game,
they do not know when to stop.  After all, ending immediately is perhaps a good choice to
solve this dilemma.

\chapitre{Related works}{Related works}
\label{sec:rel_works}

To our knowledge, the only application of coinduction to extensive game theory has been
made by \cite{capretta:2007} who uses coinduction to define only common knowledge not
equilibria in infinite games.  Another strongly connected work is this of
\citet{coupet-grimal03:_axiom_of_linear_temp_logic} on temporal logic.  Other applications
are on representation of real numbers by infinite sequences
\citep{bertot07:_affin_funct_and_series_with,DBLP:conf/flops/Julien08} and implementation
of streams (infinite lists) in electronic circuits
\citep{DBLP:journals/fac/Coupet-GrimalJ04}.  An ancestor of our description of infinite
games and infinite strategy profiles is the constructive description of finite games,
finite strategy profiles, and equilibria by \citet{vestergaard06:IPL}.
\citet{DBLP:journals/corr/abs-0904-3528} introduces the framework of infinite games.
Infinite games are introduced in \citet{osborne94:_cours_game_theory} and
\citet{osborne04a} using histories, but this is not algorithmic and therefore not amenable
to formal proofs and coinduction.  Exercice 175.1 of \citet{vestergaard06:IPL} is the
dollar auction and no infinite subgame perfect equilibrium is considered.

Many authors have studied infinite games (see for instance
\citet{martin98:_deter_of_black_game,DBLP:conf/dagstuhl/Mazala01}), but except the name
``game'' (an overloaded one), those games have nothing to see with infinite extensive
games as presented in this \booksecpoint  The infiniteness of Blackwell games for instance is
derived from a topology, by adding real numbers and probability.
\cite{DBLP:journals/toplas/Sangiorgi09} mentioned the connection between
Ehrenfeucht-Fra{\"i}ss{\'e} games \citep{EF-finite-mt} and coinduction, but the connection
with extensive games is extremely remote.

This work started after this of \cite{vestergaard06:IPL} on finite games and finite
strategy profiles.  We first developed proofs on finite strategy profiles, but unlike
Vestergaard who based his formalization on fixpoint definitions of predicates, we used
only inductive definitions of predicates. Like Vestergaard, we were able to prove the main
lemma of finite extensive games, namely that \emph{backward induction strategy profiles
  are Nash equilibria}; the script is available at
\url{http://perso.ens-lyon.fr/pierre.lescanne/COQ/INFGAMES/SCRIPTS/finite_games.v}.

Overall, this ``induction based'' presentation allowed us to switch more easily to
coinduction on infinite games.  Beside this, a development in \Coq{} of finite games with
an arbitrary number of choices at any node has been made by \citet[p. 83 and
following]{LeRouxPhD08} and an exploration of common knowledge, induction and Aumann's
theorem on rationality has been proposed by
\citet{vestergaard06:_lescan_ono}. In~\cite{lescanne07:_mechan_coq}, there is a
presentation of a somewhat connected development in \Coq, namely this of the \emph{logic
  of common knowledge}.

Since we are talking about some computational aspects of games, people may make some analogies with other works, let us state what extensive games are not.
\begin{itemize}
\item Extensive games are not \emph{semantic games} as presented in \cite{s.abramsky94:_games_and_full_compl_for,lorenz78:_dialog_logik,Locus_Solum_966909,benthem06:_logic_in_games}.

\item Extensive games are not \emph{logical games} used in proving properties of automata and protocols \cite{DBLP:conf/tphol/Merz00,DBLP:conf/provsec/AffeldtTM07}.
\item This work has only loose connection with \emph{algorithmic game theory} \cite{1296179,daskalakis09:_compl_of_comput_nash_equil}, which is more interested by the complexity of
  the algorithms, especially those which compute equilibria, than by their correction, and does not deal with infinite games.
\item Extensive games are not Ehrenfeucht-Fra\"iss\'e games, but as reported by
  \cite{DBLP:journals/toplas/Sangiorgi09} they are related to coinduction trough
  bisimilation.
\end{itemize}

The book of \cite{dowek07:_les_metam_du_calcul} gives a philosophical perspective of using
a proof assistant based on type theory in mathematics.

\chapitre{Postface}{Postface}
\label{cha:postface}

\begin{center}
  \doublebox{
    \parbox{.75\textwidth}{ In a world of \textbf{finite resources} \emph{\textbf{escalation is irrational.}}

      \medskip In a world of \textbf{infinite resources} \emph{\textbf{escalation is rational.}}}}
\end{center}

\bigskip 

Two words are more or less synonyms: escalation and (speculative) bubble which both lead
to a crash. They all yields the same outcome: a violent change in the economy, from a
growth to a sudden drop.  The main question is to know whether this attitude is
wise or whether this is a consequence of the madness of men that Newton was unable to
explain.   Rationality depends on the view agents have of the world.  This view is in terms of
availability of resources.

\ssection{Finite resources vs infinite resources}{Finite resources vs infinite resources}
\label{sec:view-finite-world}

The preceding discussion shows that the relation between rationality and escalation is
connected with the perception of the finiteness of the world and/or the (finite or infinite) quantity of
available resources.  Actually people can be split into two categories:
\begin{itemize}
\item People who view the resources as finite take escalation as irrational.
\item People who view the resources as infinite take escalation as rational.
\end{itemize}

In the first category, we can put the environmentalists, the Club of Rome, Al
Gore\footnote{This division exists among specialiste of set theory, \ie among those who
  accept the axiom of well-foundedness \citep{neumann28:_zur_theor_gesel} and those who
  reject it \citep{aczel88:_non_well_found_sets}.}  and in the second category the
speculators, the gamblers, the risk takers, the Concorde project managers, Lyndon Johnson
and Robert McNamara, Bernard Madoff, John Allen Paulos, Bernie Ebbers (WorldCom CEO),
Kim-Jong-Il, Muammar Gaddafi, Bashar al-Assad, the Greece rulers, to cite a few.  We claim
that the pros and cons of infinite resources are the same as the pros and cons of
escalation.  Perhaps, the reader like the author may think that opting for the first
category is wise, but the second option has also many fans, since escalation is everywhere
in the current world.  Those persons are not mad and have their own rationality.

\ssection{Ubiquity of escalation}{Ubiquity of escalation}

Escalation appears in many fields.

\begin{description}
\item[In economy] The most amazing example of escalation is \emph{speculative or economic
    bubble}, a well-known and old phenomenon
  \citep{kindleberger2005manias}. \citet{blanchard82:_crises} analyse the rationality of
  such bubbles, but their perpective is slightly different.  Let us assume like them, that
  rationality is the way agents take into account \emph{rational expectation} from
  uncertainty.  For us, uncertainty is non determinism or more precisely what lies between a fully
  non deterministic future and a fully stochastic one.  Since \citet{blanchard82:_crises}
  base their analysis on probability only (a full stochastic future), the accuracy of
  their approach is unsure.  In our framework, agents know only a non deterministic future
  and we explain how they reason based on this uncertain knowledge.  Their reasoning uses
  especially sequential game theory and coinduction.

  \citet{paulos03:_mathem_plays_stock_market} presents a case study on escalation, namely
  how despite being a mathematician he kept buying WorldCom stock from early 2000, when it
  was \$47 per share, till April 2002, when it was \$5.  Incidentally, he noticed that following
  an escalation process WorldCom acquired Digex in June 2001 for more than 120 times its
  actual value.

\item[In evolution theory] The \emph{red queen hypotheses} and the survival of species.
  Assume two species are living together and compete for resource.  The only way for both
  species to survive is to increase their fitness.  This lead to a kind of arm race or a
  kind of escalation.  In other words, for species to survive (in an infinite world), it
  is necessary to escalate, which is sensible in a presentation of the species competition
  by game theory.
\item[In justice] It is often the case that in court, a succession of cases and appeals lead
  to an escalation. The British case \emph{McDonald's Restaurants vs Morris \& Steel}
  \citep{vidal97:_mclib_burger_cultur_trial} is typical in this respect.  The more
  McDonald's kept suiting, the more it was loosing.
\item[Polemology] Perhaps war and conflict is where the concept of escalation was first
  developed.  Hitler's trajectory from the Beer Hall Putsch to the Battle of Stalingrad
  and eventually to his suicide was a typical escalation.  Muammar Gaddafi and Bashar
  al-Assad Are contemporary examples.

\end{description}

\ssection{Cogntive psychology}{Escalation and cognitive psychology}
\label{sec:escal-congnt-psych}

It is worth to wonder whether the agents we consider are really rational and whether they
are as rational as we would like them to be.  \citet{stanovich2010intelligence} discusses
rationality from the point of view of cognitive psychology.  Roughly
speaking a rational agent owns a mindware which refers to the rules, knowledge, procedures
and strategies that a person can retrieve from memory to aid decision making and problem
solving.  Of course we suppose that the mindware contains coinduction or at least a set of
deduction rules which covers the power of coinduction, that is which is able to reason on
infinite mathematical objects as coinduction does.

We distinguish two kinds of rationality, from the more elementary one to the more
elaborate one.  \emph{Instrumental rationality} is behaving in the world so that you get
exactly what you most want, given the resources (physical and mental) available to you.
Economist and cognitive scientists have refined the notion of optimization of goal
fulfillment into the technical notion of \emph{expected utility}.  On the other hand,
\emph{epistemic rationality} lies above instrumental rationality and interacts with it. It
tells how well the beliefs map onto the actual structure of the world.  More roughly,
epistemic rationality is about what is true and instrumental rationality is about what to
do.  The first kind of rationality corresponds to algorithmic mind and the second form of
rationality corresponds to reflexive mind.  Among others, a reflexive mind is able to
analyze the way it reasons.

An escalating agent has an \emph{algorithmic min}d, but lacks a \emph{reflexive mind},
which would allow him to change his beliefs.  In particular, he should revise his belief
in infinite resources.  Of course, at first, this gives him the power to go on and does
not inhibit him in his rush forward, but as we know `` Errare humanum est, perseverare
diabolicum''.  There is a time when he should understand that believing into an infinite
world of resources leads to a dead-end and it is better-off for him to revise his belief.
Acting so the agents shows his full rationality.

\ssection{Conclusion}{Conclusion}
\label{sec:conclusion}

Thanks to coinduction, we have reconciled human reasoning with rational reasoning in infinite extensive games.  In other words, we claim that human agents reason actually by
coinduction when faced to infinite games and are rational.  Moreover we have shown once more the threshold between finiteness and infiniteness and that reasoning on infinite objects
is not the limit when the size goes to infinity of reasoning on finite objects.



\appendix

\chapitre{Two subtle points}{Two subtle points of coinduction}

\ssection{Equalities}{Equalities}
\label{sec:eq}

\emph{Leibniz equality} says that $x=y$ if and only if, for every predicate $P$, $P(x)$ implies $P(y)$.  \emph{Extensional equality} says that $f=g$ if and only if, for all $x$,
$f(x)=g(x)$. In general, knowing a (recursive) definition of $f$ and a (recursive) definition of~$g$ is not enough to decide whether $f=g$ or $f\neq g$.  For instance, no one knows
how to prove that the two functions:
\begin{eqnarray*}
  f(1)&=&1\\
  f(2x) &=& f(x)\\
  f(2x+1) &=& f(3x+2).
\end{eqnarray*}
and
\[g(x)~=~1\] are equal, despite it is more likely that they are.  More generally, there is
no algorithm (no rigorous reasoning) which decides whether a given function $h$ is equal
to the above function $g$ or not.  Thus \emph{extensional equality} is not decidable.
Saying that two sequences that have equal elements are equal requires \emph{extensional
  equality} and it makes sense to reject such an equality when reasoning finitely about
infinite objects, like human agents would do.

\ssection{No utility in infinite runs}{Why in infinite runs, agents do not have a utility?}
\label{sec:why-infinite-plays}

In an infinite play, a play that runs forever, \ie that does not lead to
a leaf, no agent has a utility.  People might say that this an anomaly, but 
this is perfectly sensible.   In arbitrary long plays, which lead to a
leaf, all agents have a utility.  Only in plays that diverge, it is the case that agents
have no utility.  This fits well with \cite{binmore88:_model_ration_player_part_ii}
statements \emph{``The use of computing machines (automata) to model players in an
  \emph{evolutive} context is presumably uncontroversial ... machines are also appropriate
  for modeling players in an \emph{eductive} context''}.  Here we are concerned by the
eductive context where \emph{``equilibrium is achieved through careful reasoning by the
  agents before and during the play of the game''}
\cite[loc. cit]{binmore88:_model_ration_player_part_ii}.  By automaton, we mean any model
of computation\footnote{Our model of computation is this of the calculus of inductive
  construction, a kind of $`l$-calculus behind \Coq{}
  \citep{turing37:_comput_and_lambda_defin}.}, since all the models of computation are
equivalent by Church thesis.  If an agent is modeled by an automaton, this means also that
the function that computes the utility for this agent is also modeled by an automaton.  It
seems then sensible that one cannot compute the utility or the cost of an agent for an
infinite play, since computing is a finite process working on finite data (or at least
data that are finitely described).  Since the agent cannot compute the utility of an
infinite play, no sensible value can be attributed to him.  If one wants absolutely to
assign a value to an infinite play, one must abandon the automaton framework.  Moreover
this value should be the limit of a sequence of values, which does not exist in most of
the cases.\footnote{If utilities are natural numbers, it exists only if the sequence is
  stationary, which is not the case in escalation.}  For instance, in the case of the
dollar auction (Section~\ref{sec:dol}), the costs of \Alice{} associated with the infinite
play are the sequence $..., v+n, n+1, v+n+1, n+2, ...$ In the case of the $0,1$ game the payoff
of \Alice{} is $0, 1, 0, 1, ...$.  Therefore considering that in infinite plays, agents have
no utility or costs is perfectly consistent with a modeling of agents by automata.  By the
way, does an agent care about a payoff he (she) receives in infinitely many years?  Will
he (she) adapt his (her) strategy on this?

\chapitre{About \Coq{}} {About the \Coq{} development} 
\label{chap:coq}

\ssection{Functions in the \Coq{} vernacular}{The notation of functions in the \Coq{} vernacular}

In traditional mathematics, the result of applying a function $f$ to the value $x$ is
written $f(x)$ and the result of applying $f$ to $x$ and $y$ is written $f(x,y)$, this can
be considered as the result of applying $f$ to $x$ then to $y$ and written $f(x)(y)$.  In
the \Coq{} vernacular, as in type theory, one writes $f~x$ instead of $f(x)$ and $f~x~y$
instead $f(x)(y)$ or instead of $f(x,y)$ and $f~x~y~z$ instead $f(x)(y)(z)$ or $f(x,y,z)$,
because this saves parentheses and commas and because the concept of functions is the core
of the formalization.  But after all, this is just a matter of style and \Coq{} accepts
syntactic shorthands to avoid these notations when others are desirable.

\ssection{Coinduction in \Coq}{Coinduction in \Coq}
\label{sec:coq}

As we have said the proof assistant \Coq{} (\cite{Coq:manual}) plaid a central role in this research. 

\ssssection{Why should we formalize a concept in a proof assistant?}

To answer this question we like to cite Donald Knuth \citep{shustek08:_inter_donal_knuth}:
\begin{it}
  \begin{quotation}
    People have said you don't understand something until you've taught it in a class.
    The truth is you don't really understand something until you've taught it to a computer, until you've been able to program it.
  \end{quotation}
\end{it}
We claim that we can appropriately replace the last sentence by \emph{``until you've
  taught it to a proof assistant, until you've code it into
  \Coq,\footnote{\cite{Coq:manual}.}
  Isabelle,\footnote{\cite{Nipkow-Paulson-Wenzel:2002}.} or
  PVS\footnote{\cite{cade92-pvs}.}''} as it seems even more demanding to ``teach'' a proof
assistant like \Coq{} than to write a program on the same topics.  Actually without
\Coq{}, which has coinduction features, we would not have been able to capture the
concepts of Nash equilibrium and Subgame Perfect equilibria presented in this
\booksecpoint This is indeed the result of formal deduction, intuition and try and error
in \Coq{} since proving properties of infinite games and infinite strategy profiles is
extremely subtle.  Moreover by relying on a proof assistant, we can free this
\bookarticle{} from formal developments and tedious and detailed proofs, knowing anyway
that they are correct in any detail and that the reader will refer to the \Coq{} script in
case of doubt. Therefore, we can focus on informal explanation.  However, \Coq{} proposes
a readable, rigorous, and computer checked syntax, the \emph{vernacular}, for definitions,
lemmas and theorems and when we provide definitions in this \booksecvirg they are
associated with expressions stated in the vernacular provide in the appendix.  The
vernacular should be seen as a XXI$^{st}$ century version of Leibniz' \emph{characterica
  universalis} or Frege's \emph{Begriffsshrift} \citep{frege67:_from_frege_to}.

\ssssection{Decomposing an object}

The principle of coinduction is based on the greatest fixpoint of the definition, that is a \emph{coinduction defines a greatest fixpoint} (see \citep{BertotCasterant04}).  There are
two challenges when one works with such a principle: the difficulty of decomposing infinite objects and the invocation of coinduction.  They are both presented in detail
by \cite{BertotCasterant04}, but let us describe them in a few words. For the first problem, suppose one has a strategy profile $s$, which is not a leaf; one knows that $s$ is of the form
${\og{a,c, s_l,s_r}\fg}$ for some agent $a$, some choice $c$ and some strategy profiles $s_l$ and $s_r$.  To obtain such a presentation, one uses a mechanism which consists in defining a
function identity on strategy profile which is a ``clone'' of $fun\ s "=>" s\ end$:\pagebreak[3]

\medskip
\noindent
\coqdockw{Definition} \coqdocid{Strategy\_identity} (\coqdocid{s}:\coqdocid{Strategy}): \coqdocid{Strategy} :=\coqdoceol
\noindent
\coqdocid{match} \coqdocid{s} \coqdocid{with}\coqdoceol
\noindent
| $\og$\coqdocid{f}$\fg$ \ensuremath{\Rightarrow}  $\og$\coqdocid{f}$\fg$\coqdoceol
\noindent
| $\og$\coqdocid{a},\coqdocid{c},\coqdocid{sl},\coqdocid{sr}$\fg$ \ensuremath{\Rightarrow} $\og$\coqdocid{a},\coqdocid{c},\coqdocid{sl},\coqdocid{sr}$\fg$\coqdoceol
\noindent
\coqdocid{end}.\coqdoceol

\medskip

\noindent In other words, the \emph{strategy identity} function is, computationally speaking, the function which associates $\og$\coqdocid{f}$\fg$ with $\og$\coqdocid{f}$\fg$ and
$\og$\coqdocid{a},\coqdocid{c},\coqdocid{sl},\coqdocid{sr}$\fg$ with $\og$\coqdocid{a},\coqdocid{c},\coqdocid{sl},\coqdocid{sr}$\fg$ and not the function which associates $s$ with $s$.
We can prove the \emph{strategy decomposition} lemma:

\medskip
\noindent
\coqdockw{Lemma} \coqdocid{Strategy\_decomposition}: \ensuremath{\forall} \coqdocid{s}: \coqdocid{Strategy},\coqdoceol
\coqdocindent{1.00em}
\coqdocid{Strategy\_identity} \coqdocid{s} = \coqdocid{s}.\coqdoceol

\medskip

Thus when one wants to decompose a strategy $s$, one replaces $s$ by \coqdocid{Strategy\_identity} $s$ and one simplifies the expression, and one gets ${\og a,c, s_l,s_r\fg}$ for
some $a$, $c$, $s_l$ and $s_r$.

\ssssection{ Invoking coinduction}

The \emph{principle of coinduction} is based on a \emph{tactic}\footnote{A \emph{tactic} is a tool in \Coq{} used to build proofs without using the most elementary constructions.}
called \emph{cofix}.  It consists in assuming the proposition one wants to proof, provided one applies it only on strict sub-objects. In the current implementation of \Coq{}, the
user has to ensure that he invokes it on ``strict'' sub-objects.  This is not always completely trivial and requires a good methodology.  However the \emph{proof checker} (a piece
of software which accepts only correct proofs) verifies that this constraint is fulfilled at the time of checking the proof.

\ssection{Excerpts of the \Coq{} development}{Excerpts of the \Coq{} development}
\label{sec:excerpts-coq-devel}

The full development is  in the url \\\centerline{\url{http://perso.ens-lyon.fr/pierre.lescanne/COQ/ER/SCRIPTS/}}
with a description in
\\\centerline{\url{http://perso.ens-lyon.fr/pierre.lescanne/COQ/ER/HTML/}.}   

\sssection{Infinite binary trees}{Infinite binary trees}

\noindent
\coqdockw{CoInductive} \coqdocid{InfFinBintree} : \coqdocid{Set} :=\coqdoceol
\noindent
| \coqdocid{InfFinBtNil}: \coqdocid{InfFinBintree} \coqdoceol
\noindent
| \coqdocid{InfFinBtNode}: \coqdocid{InfFinBintree} \ensuremath{\rightarrow} \coqdocid{InfFinBintree} \ensuremath{\rightarrow} \coqdocid{InfFinBintree}.\coqdoceol

\medskip
\noindent
\coqdockw{CoInductive} \coqdocid{InfiniteInfFinBT}: \coqdocid{InfFinBintree} \ensuremath{\rightarrow} \coqdocid{Prop} := \coqdoceol
\noindent
| \coqdocid{IBTLeft} : \ensuremath{\forall} \coqdocid{bl} \coqdocid{br}, \coqdocid{InfiniteInfFinBT} \coqdocid{bl} \ensuremath{\rightarrow} \coqdocid{InfiniteInfFinBT} (\coqdocid{InfFinBtNode} \coqdocid{bl} \coqdocid{br})\coqdoceol
\noindent
| \coqdocid{IBTRight} : \ensuremath{\forall} \coqdocid{bl} \coqdocid{br}, \coqdocid{InfiniteInfFinBT} \coqdocid{br} \ensuremath{\rightarrow} \coqdocid{InfiniteInfFinBT} (\coqdocid{InfFinBtNode} \coqdocid{bl} \coqdocid{br}).\coqdoceol

\medskip
\noindent
\coqdockw{CoFixpoint} \coqdocid{Zig}: \coqdocid{InfFinBintree} := \coqdocid{InfFinBtNode} \coqdocid{Zag} \coqdocid{InfFinBtNil}\coqdoceol
\noindent
\coqdocid{with} \coqdocid{Zag}: \coqdocid{InfFinBintree} := \coqdocid{InfFinBtNode} \coqdocid{InfFinBtNil} \coqdocid{Zig}.\coqdoceol

\sssection{Infinite games}{Infinite games}

\noindent
\coqdockw{CoInductive} \coqdocid{Game} : \coqdocid{Set} :=\coqdoceol
\noindent
| \coqdocid{gLeaf}:  \coqdocid{Utility\_fun} \ensuremath{\rightarrow} \coqdocid{Game}\coqdoceol
\noindent
| \coqdocid{gNode} : \coqdocid{Agent} \ensuremath{\rightarrow} \coqdocid{Game} \ensuremath{\rightarrow} \coqdocid{Game} \ensuremath{\rightarrow} \coqdocid{Game}.\coqdoceol

\medskip

\sssection{Infinite strategy profiles}{Infinite strategy profiles}

\noindent
\coqdockw{CoInductive}
 \coqdocid{StratProf} : \coqdocid{Set} :=\coqdoceol
\noindent
| \coqdocid{sLeaf} : \coqdocid{Utility\_fun} \ensuremath{\rightarrow} \coqdocid{StratProf}\coqdoceol
\noindent
| \coqdocid{sNode} : \coqdocid{Agent} \ensuremath{\rightarrow} \coqdocid{Choice} \ensuremath{\rightarrow} \coqdocid{StratProf} \ensuremath{\rightarrow} \coqdocid{StratProf} \ensuremath{\rightarrow} \coqdocid{StratProf}.\coqdoceol

\medskip

\noindent
\coqdockw{Inductive} \coqdocid{s2u} : \coqdocid{StratProf} \ensuremath{\rightarrow} \coqdocid{Agent} \ensuremath{\rightarrow} \coqdocid{Utility} \ensuremath{\rightarrow} \coqdocid{Prop} :=\coqdoceol
\noindent
| \coqdocid{s2uLeaf}: \ensuremath{\forall} \coqdocid{a} \coqdocid{f}, \coqdocid{s2u} ($\og$ \coqdocid{f}$\fg$) \coqdocid{a} (\coqdocid{f} \coqdocid{a})\coqdoceol
\noindent
| \coqdocid{s2uLeft}: \ensuremath{\forall}  (\coqdocid{a} \coqdocid{a'}:\coqdocid{Agent}) (\coqdocid{u}:\coqdocid{Utility}) (\coqdocid{sl} \coqdocid{sr}:\coqdocid{StratProf}),\coqdoceol
\coqdocindent{2.00em}
\coqdocid{s2u} \coqdocid{sl} \coqdocid{a} \coqdocid{u}  \ensuremath{\rightarrow} \coqdocid{s2u} ($\og$ \coqdocid{a'},\coqdocid{l},\coqdocid{sl},\coqdocid{sr}$\fg$) \coqdocid{a} \coqdocid{u}  \coqdoceol
\noindent
| \coqdocid{s2uRight}: \ensuremath{\forall} (\coqdocid{a} \coqdocid{a'}:\coqdocid{Agent}) (\coqdocid{u}:\coqdocid{Utility}) (\coqdocid{sl} \coqdocid{sr}:\coqdocid{StratProf}),\coqdoceol
\coqdocindent{2.00em}
\coqdocid{s2u} \coqdocid{sr} \coqdocid{a} \coqdocid{u} \ensuremath{\rightarrow} \coqdocid{s2u} ($\og$ \coqdocid{a'},\coqdocid{r},\coqdocid{sl},\coqdocid{sr}$\fg$) \coqdocid{a} \coqdocid{u}.\coqdoceol

\medskip

\noindent
\coqdockw{Lemma} \coqdocid{Existence\_s2u}:  \ensuremath{\forall} (\coqdocid{a}:\coqdocid{Agent}) (\coqdocid{s}:\coqdocid{StratProf}),\coqdoceol
\coqdocindent{1.00em}
\coqdocid{LeadsToLeaf} \coqdocid{s} \ensuremath{\rightarrow} \ensuremath{\exists} \coqdocid{u}:\coqdocid{Utility}, \coqdocid{s2u} \coqdocid{s} \coqdocid{a} \coqdocid{u}.\coqdoceol

\medskip
\noindent
\coqdockw{Lemma} \coqdocid{Uniqueness\_s2u}: \ensuremath{\forall} (\coqdocid{a}:\coqdocid{Agent}) (\coqdocid{u} \coqdocid{v}:\coqdocid{Utility}) (\coqdocid{s}:\coqdocid{StratProf}),\coqdoceol
\coqdocindent{1.50em}
\coqdocid{LeadsToLeaf} \coqdocid{s} \ensuremath{\rightarrow} \coqdocid{s2u} \coqdocid{s} \coqdocid{a} \coqdocid{u} \ensuremath{\rightarrow} \coqdocid{s2u} \coqdocid{s} \coqdocid{a} \coqdocid{v} \ensuremath{\rightarrow} \coqdocid{u}=\coqdocid{v}.\coqdoceol

\medskip
\noindent
\coqdockw{Inductive} \coqdocid{LeadsToLeaf}: \coqdocid{StratProf} \ensuremath{\rightarrow} \coqdocid{Prop} :=\coqdoceol
\noindent
| \coqdocid{LtLLeaf}: \ensuremath{\forall} \coqdocid{f}, \coqdocid{LeadsToLeaf} ($\og$ \coqdocid{f}$\fg$)\coqdoceol
\noindent
| \coqdocid{LtLLeft}: \ensuremath{\forall} (\coqdocid{a}:\coqdocid{Agent})(\coqdocid{sl}: \coqdocid{StratProf}) (\coqdocid{sr}:\coqdocid{StratProf}),\coqdoceol
\coqdocindent{2.00em}
\coqdocid{LeadsToLeaf} \coqdocid{sl} \ensuremath{\rightarrow} \coqdocid{LeadsToLeaf} ($\og$ \coqdocid{a},\coqdocid{l},\coqdocid{sl},\coqdocid{sr}$\fg$)\coqdoceol
\noindent
| \coqdocid{LtLRight}: \ensuremath{\forall} (\coqdocid{a}:\coqdocid{Agent})(\coqdocid{sl}: \coqdocid{StratProf}) (\coqdocid{sr}:\coqdocid{StratProf}),\coqdoceol
\coqdocindent{2.50em}
\coqdocid{LeadsToLeaf} \coqdocid{sr} \ensuremath{\rightarrow} \coqdocid{LeadsToLeaf} ($\og$ \coqdocid{a},\coqdocid{r},\coqdocid{sl},\coqdocid{sr}$\fg$).\coqdoceol

\medskip
\noindent
\coqdockw{CoInductive} \coqdocid{AlwLeadsToLeaf}: \coqdocid{StratProf} \ensuremath{\rightarrow} \coqdocid{Prop} :=\coqdoceol
\noindent
| \coqdocid{ALtLeaf} : \ensuremath{\forall} (\coqdocid{f}:\coqdocid{Utility\_fun}), \coqdocid{AlwLeadsToLeaf} ($\og$\coqdocid{f}$\fg$)\coqdoceol
\noindent
| \coqdocid{ALtL} : \ensuremath{\forall} (\coqdocid{a}:\coqdocid{Agent})(\coqdocid{c}:\coqdocid{Choice})(\coqdocid{sl} \coqdocid{sr}:\coqdocid{StratProf}),\coqdoceol
\coqdocindent{2.00em}
\coqdocid{LeadsToLeaf} ($\og$\coqdocid{a},\coqdocid{c},\coqdocid{sl},\coqdocid{sr}$\fg$) \ensuremath{\rightarrow} \coqdocid{AlwLeadsToLeaf} \coqdocid{sl} \ensuremath{\rightarrow}\coqdocid{AlwLeadsToLeaf} \coqdocid{sr} \ensuremath{\rightarrow} \coqdoceol
\coqdocindent{2.00em}
\coqdocid{AlwLeadsToLeaf} ($\og$\coqdocid{a},\coqdocid{c},\coqdocid{sl},\coqdocid{sr}$\fg$).\coqdoceol

\sssection{Convertibility}{Convertibility}
\label{sec:Coq-conv}

\noindent
\coqdockw{Inductive} \coqdocid{IndAgentConv}: \coqdocid{Agent} \ensuremath{\rightarrow} \coqdocid{StratProf} \ensuremath{\rightarrow} \coqdocid{StratProf} \ensuremath{\rightarrow} \coqdocid{Prop} :=\coqdoceol
\noindent
| \coqdocid{ConvRefl}: \ensuremath{\forall} (\coqdocid{a}:\coqdocid{Agent})(\coqdocid{s}: \coqdocid{StratProf}), \coqdocid{IndAgentConv} \coqdocid{a} \coqdocid{s} \coqdocid{s}\coqdoceol
\noindent
| \coqdocid{ConvAgent} :  \ensuremath{\forall} (\coqdocid{a}:\coqdocid{Agent})(\coqdocid{c} \coqdocid{c'}:\coqdocid{Choice})(\coqdocid{sl} \coqdocid{sl'} \coqdocid{sr} \coqdocid{sr'}:\coqdocid{StratProf}),\coqdoceol
\coqdocindent{1.00em}
(\coqdocid{IndAgentConv} \coqdocid{a} \coqdocid{sl} \coqdocid{sl'}) \ensuremath{\rightarrow} (\coqdocid{IndAgentConv} \coqdocid{a} \coqdocid{sr} \coqdocid{sr'}) \ensuremath{\rightarrow} \coqdoceol
\coqdocindent{1.00em}
\coqdocid{IndAgentConv} \coqdocid{a} ($\og$\coqdocid{a},\coqdocid{c},\coqdocid{sl},\coqdocid{sr}$\fg$) ($\og$\coqdocid{a},\coqdocid{c'},\coqdocid{sl'},\coqdocid{sr'}$\fg$)\coqdoceol
\noindent
| \coqdocid{ConvChoice} :  \ensuremath{\forall} (\coqdocid{a} \coqdocid{a'}:\coqdocid{Agent}) (\coqdocid{c}: \coqdocid{Choice}) (\coqdocid{sl} \coqdocid{sl'} \coqdocid{sr} \coqdocid{sr'}:\coqdocid{StratProf}),\coqdoceol
\coqdocindent{1.00em}
\coqdocid{IndAgentConv} \coqdocid{a} \coqdocid{sl} \coqdocid{sl'} \ensuremath{\rightarrow} (\coqdocid{IndAgentConv} \coqdocid{a} \coqdocid{sr} \coqdocid{sr'}) \ensuremath{\rightarrow} \coqdoceol
\coqdocindent{1.00em}
\coqdocid{IndAgentConv} \coqdocid{a} ($\og$\coqdocid{a'},\coqdocid{c},\coqdocid{sl},\coqdocid{sr}$\fg$) ($\og$\coqdocid{a'},\coqdocid{c},\coqdocid{sl'},\coqdocid{sr'}$\fg$).\coqdoceol

\medskip
 \coqdockw{Notation} "\coqdocid{sl}  \conva\ \coqdocid{sr}" := (\coqdocid{IndAgentConv} \coqdocid{a} \coqdocid{sl} \coqdocid{sr}). \coqdoceol

\medskip

\sssection{SGPE}{SGPE}

  \noindent \coqdockw{CoInductive} \coqdocid{SGPE}: \coqdocid{StratProf} \ensuremath{\rightarrow} \coqdocid{Prop} :=\coqdoceol
  \noindent | \coqdocid{SGPE\_leaf}: \ensuremath{\forall} \coqdocid{f}:\coqdocid{Utility\_fun}, \coqdocid{SGPE} ($\og$\coqdocid{f}$\fg$)\coqdoceol
  \noindent | \coqdocid{SGPE\_left}: \ensuremath{\forall} (\coqdocid{a}:\coqdocid{Agent})(\coqdocid{u} \coqdocid{v}: \coqdocid{Utility}) (\coqdocid{sl} \coqdocid{sr}:
  \coqdocid{StratProf}), \coqdoceol \coqdocindent{2.00em} \coqdocid{AlwLeadsToLeaf} ($\og$\coqdocid{a},\coqdocid{l},\coqdocid{sl},\coqdocid{sr}$\fg$) \ensuremath{\rightarrow}
  \coqdoceol \coqdocindent{2.00em} \coqdocid{SGPE} \coqdocid{sl} \ensuremath{\rightarrow} \coqdocid{SGPE} \coqdocid{sr} \ensuremath{\rightarrow} \coqdoceol \coqdocindent{2.00em}
  \coqdocid{s2u} \coqdocid{sl} \coqdocid{a} \coqdocid{u} \ensuremath{\rightarrow} \coqdocid{s2u} \coqdocid{sr} \coqdocid{a} \coqdocid{v} \ensuremath{\rightarrow} (\coqdocid{v}
  \leut \coqdocid{u}) \ensuremath{\rightarrow} \coqdoceol \coqdocindent{2.00em} \coqdocid{SGPE} ($\og$\coqdocid{a},\coqdocid{l},\coqdocid{sl},\coqdocid{sr}$\fg$)\coqdoceol
  \noindent | \coqdocid{SGPE\_right}: \ensuremath{\forall} (\coqdocid{a}:\coqdocid{Agent}) (\coqdocid{u} \coqdocid{v}:\coqdocid{Utility}) (\coqdocid{sl} \coqdocid{sr}:
  \coqdocid{StratProf}), \coqdoceol \coqdocindent{2.00em} \coqdocid{AlwLeadsToLeaf} ($\og$\coqdocid{a},\coqdocid{r},\coqdocid{sl},\coqdocid{sr}$\fg$) \ensuremath{\rightarrow}
  \coqdoceol \coqdocindent{2.00em} \coqdocid{SGPE} \coqdocid{sl} \ensuremath{\rightarrow} \coqdocid{SGPE} \coqdocid{sr} \ensuremath{\rightarrow} \coqdoceol \coqdocindent{2.00em}
  \coqdocid{s2u} \coqdocid{sl} \coqdocid{a} \coqdocid{u} \ensuremath{\rightarrow} \coqdocid{s2u} \coqdocid{sr} \coqdocid{a} \coqdocid{v} \ensuremath{\rightarrow} (\coqdocid{u}
  \leut \coqdocid{v}) \ensuremath{\rightarrow} \coqdoceol \coqdocindent{2.00em} \coqdocid{SGPE} ($\og$\coqdocid{a},\coqdocid{r},\coqdocid{sl},\coqdocid{sr}$\fg$). \coqdoceol

\sssection{Nash equilibrium}{Nash equilibrium}

\noindent
\coqdockw{Definition} \coqdocid{NashEq} (\coqdocid{s}: \coqdocid{StratProf}): \coqdocid{Prop} := \coqdoceol
\coqdocindent{1.00em}
\ensuremath{\forall} \coqdocid{a} \coqdocid{s'} \coqdocid{u} \coqdocid{u'}, \coqdocid{s'}\conva\coqdocid{s} \ensuremath{\rightarrow} 
(\coqdocid{s2u} \coqdocid{s'} \coqdocid{a} \coqdocid{u'}) 
\ensuremath{\rightarrow} (\coqdocid{s2u} \coqdocid{s} \coqdocid{a} \coqdocid{u}) \ensuremath{\rightarrow} 
(\coqdocid{u'} \leut \coqdocid{u}).\coqdoceol

\sssection{Dollar Auction}{Dollar Auction}
\label{sec:dollar-auction}

\noindent
\coqdockw{Notation} "[ x , y ]" := \coqdoceol
\coqdocindent{1.00em}
(\coqdocid{sLeaf} (\coqdocid{fun} \coqdocid{a}:\coqdocid{Alice\_Bob} \ensuremath{\Rightarrow} \coqdocid{match} \coqdocid{a} \coqdocid{with}
\coqdocid{Alice} \ensuremath{\Rightarrow} \coqdocid{x} $\mid$ \coqdocid{Bob} \ensuremath{\Rightarrow} \coqdocid{y} \coqdocid{end})) 

(\coqdocid{at} \coqdocid{level} 80).\coqdoceol

\sssection{{\Alice} stops always and \Bob{} continues always}{{\Alice} stops always and \Bob{} continues always}

\noindent
\coqdockw{Definition} \coqdocid{add\_Alice\_Bob\_dol} (\coqdocid{cA} \coqdocid{cB}:\coqdocid{Choice}) (\coqdocid{n}:\coqdocid{nat}) (\coqdocid{s}:\coqdocid{Strat}) :=\coqdoceol
\coqdocindent{1.00em}
\og\coqdocid{Alice},\coqdocid{cA},\og\coqdocid{Bob}, \coqdocid{cB},\coqdocid{s},[\coqdocid{n}+1, \coqdocid{v}+\coqdocid{n}]\fg,[\coqdocid{v}+\coqdocid{n},\coqdocid{n}]\fg.\coqdoceol

\medskip
\noindent
\coqdockw{CoFixpoint} \coqdocid{dolAcBs} (\coqdocid{n}:\coqdocid{nat}): \coqdocid{Strat} := \coqdocid{add\_Alice\_Bob\_dol} \coqdocid{l} \coqdocid{r} \coqdocid{n} (\coqdocid{dolAcBs} (\coqdocid{n}+1)).\coqdoceol

\medskip
\noindent
\coqdockw{Theorem} \coqdocid{SGPE\_dol\_Ac\_Bs}:  \ensuremath{\forall} (\coqdocid{n}:\coqdocid{nat}), \coqdocid{SGPE} \coqdocid{ge} (\coqdocid{dolAcBs} \coqdocid{n}).\coqdoceol

\sssection{\Alice{} continues always and \Bob{} stops always}{\Alice{} continues always and \Bob{} stops always}

\noindent
\coqdockw{CoFixpoint} \coqdocid{dolAsBc} (\coqdocid{n}:\coqdocid{nat}): \coqdocid{Strat} := \coqdocid{add\_Alice\_Bob\_dol} \coqdocid{r} \coqdocid{l} \coqdocid{n} (\coqdocid{dolAsBc} (\coqdocid{n}+1)).\coqdoceol

\medskip
\noindent
\coqdockw{Theorem} \coqdocid{SGPE\_dol\_As\_Bc}:  \ensuremath{\forall} (\coqdocid{n}:\coqdocid{nat}), \coqdocid{SGPE} \coqdocid{ge} (\coqdocid{dolAsBc} \coqdocid{n}).\coqdoceol

\sssection{Always give up}{Always give up}

\noindent
\coqdockw{CoFixpoint} \coqdocid{dolAsBs} (\coqdocid{n}:\coqdocid{nat}): \coqdocid{Strat} := \coqdocid{add\_Alice\_Bob\_dol} \coqdocid{r} \coqdocid{r} \coqdocid{n} (\coqdocid{dolAsBs} (\coqdocid{n}+1)).\coqdoceol

\medskip
\noindent
\coqdockw{Theorem} \coqdocid{NotSGPE\_dolAsBs}: (\coqdocid{v}$>$1) \ensuremath{\rightarrow} \~{}(\coqdocid{NashEq} \coqdocid{ge} (\coqdocid{dolAsBs} 0)).\coqdoceol

\sssection{Infinipede}{Infinipede}
\label{sec:infinipede}

\noindent
\coqdockw{Definition} \coqdocid{add\_Alice\_Bob\_cent} (\coqdocid{cA} \coqdocid{cB}:\coqdocid{Choice}) (\coqdocid{n}:\coqdocid{nat}) (\coqdocid{s}:\coqdocid{Strat}) :=\coqdoceol
\coqdocindent{1.00em}
$\og$\coqdocid{Alice},\coqdocid{cA},$\og$\coqdocid{Bob}, \coqdocid{cB},\coqdocid{s},[2\ensuremath{\times}\coqdocid{n}-1, 2\ensuremath{\times}\coqdocid{n}+3]$\fg$,[2\ensuremath{\times}\coqdocid{n},2\ensuremath{\times}\coqdocid{n}]$\fg$.\coqdoceol

\noindent
\coqdockw{CoFixpoint} \coqdocid{cent\_agu} (\coqdocid{n}:\coqdocid{nat}): (\coqdocid{Strat}) := 

\coqdocid{add\_Alice\_Bob\_cent} \coqdocid{r} \coqdocid{r} \coqdocid{n} (\coqdocid{cent\_agu} (\coqdocid{S} \coqdocid{n})).\coqdoceol

\medskip
\noindent
\coqdockw{Lemma} \coqdocid{AlwLeadsToLeaf\_cent\_agu}:  \ensuremath{\forall} (\coqdocid{n}:\coqdocid{nat}), \coqdocid{AlwLeadsToLeaf} (\coqdocid{cent\_agu} \coqdocid{n}).\coqdoceol

\medskip
\noindent
\coqdockw{Lemma} \coqdocid{LeadsToLeaf\_cent\_agu}: \ensuremath{\forall} (\coqdocid{n}:\coqdocid{nat}), \coqdocid{LeadsToLeaf} (\coqdocid{cent\_agu} \coqdocid{n}).\coqdoceol

\medskip
\noindent
\coqdockw{Theorem} \coqdocid{SGPE\_cent\_AGU}:  \ensuremath{\forall} (\coqdocid{n}:\coqdocid{nat}), \coqdocid{SGPE} \coqdocid{le} (\coqdocid{cent\_agu} \coqdocid{n}).\coqdoceol

\medskip
\noindent
\coqdockw{Lemma} \coqdocid{NashEq\_cent\_agu}: \ensuremath{\forall} (\coqdocid{n}:\coqdocid{nat}), \coqdocid{NashEq} \coqdocid{le} (\coqdocid{cent\_agu} \coqdocid{n}).\coqdoceol

\noindent
\sssection{Escalation}{Escalation}

\noindent
\coqdockw{Definition} \coqdocid{has\_an\_escalation\_sequence} (\coqdocid{g\_seq}:\coqdocid{nat} \ensuremath{\rightarrow} \coqdocid{Game}): \coqdocid{Prop} := \ensuremath{\forall} \coqdocid{n}:\coqdocid{nat}, \coqdoceol
\noindent
\ensuremath{\exists} \coqdocid{s},  \ensuremath{\exists} \coqdocid{s'}, \ensuremath{\exists} \coqdocid{a}, \coqdoceol
\coqdocindent{1.00em}
(\coqdocid{s2g} (\og\coqdocid{a},\coqdocid{l},\coqdocid{s},\coqdocid{s'}\fg) =\coqdocid{gbis}= \coqdocid{g\_seq} \coqdocid{n} \ensuremath{\land} \coqdocid{SGPE} (\og\coqdocid{a},\coqdocid{l},\coqdocid{s},\coqdocid{s'}\fg) \ensuremath{\lor}\coqdoceol
\coqdocindent{1.00em}
\coqdocid{s2g} (\og\coqdocid{a},\coqdocid{r},\coqdocid{s'},\coqdocid{s}\fg) =\coqdocid{gbis}= \coqdocid{g\_seq} \coqdocid{n} \ensuremath{\land} \coqdocid{SGPE} (\og\coqdocid{a},\coqdocid{r},\coqdocid{s'},\coqdocid{s}\fg)) \ensuremath{\land}\coqdoceol
\coqdocindent{1.00em}
\coqdocid{s2g} \coqdocid{s} =\coqdocid{gbis}= \coqdocid{g\_seq} (\coqdocid{n}+1).\coqdoceol

\medskip
\noindent
\coqdockw{Definition} \coqdocid{has\_an\_escalation} (\coqdocid{g}:\coqdocid{Game}) : \coqdocid{Prop} :=\coqdoceol
\coqdocindent{2.00em}
\ensuremath{\exists} \coqdocid{g\_seq}, (\coqdocid{has\_an\_escalation\_sequence} \coqdocid{g\_seq}) \ensuremath{\land} (\coqdocid{g\_seq} 0 = \coqdocid{g}).\coqdoceol
\coqdocindent{1.00em}
\coqdoceol

\medskip
\noindent
\coqdockw{Theorem} \coqdocid{Zero\_one\_game\_has\_an\_escalation}: \coqdoceol
\coqdocindent{2.00em}
\coqdocid{has\_an\_escalation} \coqdocid{Alice\_Bob} \coqdocid{nat} \coqdocid{ge}
\coqdocid{zero\_one\_game}.\coqdoceol

\medskip
\noindent
\coqdockw{Theorem} \coqdocid{Dollar\_game\_has\_an\_escalation}: \coqdoceol
\coqdocindent{2.00em}
\coqdocid{has\_an\_escalation} \coqdocid{Alice\_Bob} \coqdocid{nat} \coqdocid{ge} (\coqdocid{dollar\_game} 0).\coqdoceol




\ifBook\else
\begin{landscape}
\begin{figure}
  \centering
     \vspace*{60pt}
  \begin{tiny}
    \prooftree \prooftree \prooftree \justifies\og{ \Alice \mapsto 0, \Bob \mapsto 1}\fg
    \convalice \og{\Alice \mapsto 0, \Bob \mapsto 1}\fg \using \textit{ConvRefl}
    \endprooftree
    \ \ \prooftree \justifies\og{ \ldots}\fg \convalice \og{ \ldots}\fg \using
    \textit{ConvRefl}
    \endprooftree
    \justifies \og \emph{\Bob}, \lft, \og{ \ldots}\fg \og{ \ldots}\fg \fg%
    \convalice %
    \og\emph{\Bob}, \lft, \og\ldots\fg \og{ \ldots}\fg \fg \using \textit{ConvAgent}
    \endprooftree
    \ \ \ \prooftree \justifies \og{\Alice \mapsto 2, \Bob \mapsto 0}\fg \convalice
    \og{\Alice \mapsto 2, \Bob \mapsto 0}\fg \using \textit{ConvRefl}
    \endprooftree
    \justifies \og\emph{\Alice},\lft,\og\emph{\Bob}, \lft, \og{ \ldots}\fg, %
    \og{\ldots}\fg \fg, \og \ldots \fg\fg \convalice \og\emph{\Alice},\rgt,\og\emph{\Bob},
    \lft, \og{ \ldots}\fg, \og{\ldots}\fg \fg, \og \ldots \fg\fg \using
    \textit{ConvChoice}
    \endprooftree
  \end{tiny}
  \caption{An inductive proof of convertibility}
  \label{fig:ind_proof_conv}
\end{figure}
\end{landscape}
\fi

\printindex

\bibliographystyle{plainnat}

\end{document}